\documentclass[submission,copyright,creativecommons]{eptcs}
\usepackage{xspace}
\usepackage{amsthm,amsmath,amssymb,nicefrac}
\usepackage{latexsym}
\usepackage{txfonts}
\usepackage{wrapfig}
\usepackage{tikz}
\usepackage{mathtools}
\usepackage[font=small,skip=0pt]{caption}
\usetikzlibrary{matrix,arrows,automata,decorations.pathreplacing,shapes,shapes.geometric,calc,fit,positioning}

\title{Multi-Buffer Simulations for Trace Language Inclusion}
\def\titlerunning{Multi-Buffer Simulations}
\author{Milka Hutagalung$^1$ \qquad Norbert Hundeshagen$^1$ 
\qquad Dietrich Kuske$^2$ \\ Martin Lange$^1$ \qquad Etienne Lozes$^3$
\institute{$^1$ School of Electrical Engineering and Computer Science 
University of Kassel, Germany \\
$^2$ Technische Universit\"at Ilmenau, Germany \\
$^3$ LSV, ENS Cachan, France}}

\def\authorrunning{M.~Hutagalung, N.~Hundeshagen, D.~Kuske, {} \\ M.~Lange \&  E.~Lozes}

\theoremstyle{plain}
\newtheorem{theorem}{Theorem}

\newtheorem{corollary}[theorem]{Corollary}
\newtheorem{lemma}[theorem]{Lemma}

\theoremstyle{definition}

\newtheorem{example}[theorem]{Example}

\theoremstyle{remark}
\newtheorem{remark}[theorem]{Remark}

\newcommand{\spoiler}{\textsc{Spoiler}\xspace}
\newcommand{\duplicator}{\textsc{Duplicator}\xspace}
\newcommand{\Nat}{\mathbb{N}}

\def\newarrow#1{\mathop{{\hbox{\setbox0=\hbox{$\scriptstyle{#1\quad}$}{$%
\mathrel{\mathop{\setbox1=\hbox to
\wd0{\rightarrowfill}\ht1=3pt\dp1=-2pt\box1}\limits^{#1}}%
$}}}}}

\renewcommand{\epsilon}{\varepsilon}

\newcommand{\Transition}[3]{\ensuremath{#1 \newarrow{#2} #3}}

\newcommand{\DKsimugame}[4]{\ensuremath{\mathcal{G}^{#1}_{#2}(#3,#4)}}
\newcommand{\delaygame}[3]{\ensuremath{\mathcal{G}^{#1}_{\mathsf{del}}(#2,#3)}}

\newcommand{\DKsimu}[2]{\ensuremath{\mathbin{\sqsubseteq^{#1}_{#2}}}}
\newcommand{\DKnotsimu}[2]{\ensuremath{\mathbin{\nsqsubseteq^{#1}_{#2}}}}

\begin{document}
\maketitle

\begin{abstract}
  We consider simulation games played between Spoiler and Duplicator
  on two B\"uchi automata in which the choices made by Spoiler can be
  buffered by Duplicator in several buffers before she executes them
  on her structure. We show that the simulation games are useful to approximate the
  inclusion of trace closures of languages accepted by finite-state
  automata, which is known to be undecidable. We study the
  decidability and complexity and show that the game with bounded
  buffers can be decided in polynomial time, whereas the game with one
  unbounded and one bounded buffer is highly undecidable. We also show
  some sufficient conditions on the automata for Duplicator to win the
  game (with unbounded buffers).
\end{abstract}

\vspace*{-7mm}

\section{Introduction}

Simulation is a pre-order between labeled transition systems
$\mathcal{T}$ and $\mathcal{T}'$ that formalizes the idea that
``$\mathcal{T}'$ can do everything that $\mathcal{T}$ can''. Formally,
it relates the states of the two transition systems such that each
state $t'$ in $\mathcal{T}'$ that is connected to some $t$ in
$\mathcal{T}$ can mimic the immediate behaviour of $t$, i.e.\ it
carries the same label, and whenever $t$ has a successor then $t'$ has
a matching one, too.

Simulation relations have become popular in the area of automata
theory because they can be used to efficiently under-approximate
language inclusion problems for automata on finite or infinite words
and trees and to minimise such automata \cite{lncs575*255,
  EtessamiWS01, journals/tcs/FritzW05,
  conf/tacas/AbdullaBHKV08}.
One advantage of these simulation relationships is that
they are often computable in polynomial time whereas language
inclusion problems are PSPACE-complete for typical (finite, B\"uchi,
parity, etc.) automata on words and EXPTIME-complete for such automata
on trees. To reason about simulation relations, one very often
characterises them by the existence of winning strategies of the
second player in certain two-player games. These are
played on the state spaces of two automata where one player (\spoiler)
reveals a run of the first automaton piece-wise and the second player
(\duplicator) has to produce a corresponding run of the second
automaton (where ``corresponding'' often means ``on the same word or
tree''). The simplest such game requires \spoiler to produce one step
of his run per round and \duplicator to answer immediately by one step
of her run. With this game, it is very easy to construct pairs of
automata such that language inclusion holds but simulation does not
(i.e., \duplicator has no winning strategy). Intuitively, \duplicator
is too weak to capture language inclusion. This observation has led to
the study of several extensions of simulation relations and games with
the aim of making \duplicator stronger or \spoiler weaker whilst
retaining a better complexity than language inclusion. Examples in
this context are \emph{multi-pebble simulation}
\cite{conf/concur/Etessami02}, \emph{multi-letter simulation}
\cite{HutagalungLL:LATA13,MayrC13}, \emph{buffered simulations}
\cite{DBLP:journals/corr/HutagalungLL14}, and \emph{delayed games}
\cite{DBLP:journals/corr/abs-1209-0800}.  In all these contexts, the
winning condition is a regular set of infinite words over the set of
pairs of letters (this is explicit in
\cite{DBLP:journals/corr/abs-1209-0800} and implicit in
\cite{conf/concur/Etessami02,HutagalungLL:LATA13,MayrC13,
  DBLP:journals/corr/HutagalungLL14} where \duplicator aims to produce
the same word).

In this paper, we aim at approximating the inclusion of the
Mazurkiewicz trace closure of two regular languages using simulation technology. 
More precisely, we are given two B\"uchi automata $\mathcal A$ and $\mathcal B$ and a trace alphabet and we ask whether,
for every infinite word accepted by $\mathcal A$, there is a trace-equivalent word accepted by $\mathcal B$. This problem was shown to be
undecidable by Sakarovitch~\cite{Sak92} (and
\cite{journals/ijfcs/Finkel12} can be used to prove that it is even
highly undecidable). To approximate this problem, we use a game
approach as indicated above, i.e., \spoiler and \duplicator reveal
runs of $\mathcal A$ and $\mathcal B$ piece-wise producing, in the
limit, a pair of runs. In doing so, \duplicator tries to produce a run
on a trace-equivalent word. Since the set of pairs of trace-equivalent
words is not regular, \duplicator's winning condition is not a regular
set. Hence the results from
\cite{DBLP:journals/corr/abs-1209-0800,conf/concur/Etessami02,
  HutagalungLL:LATA13,MayrC13, DBLP:journals/corr/HutagalungLL14} are
not applicable.

To overcome this problem, we first restrict \duplicator's moves in
such a way that she is forced to produce a prefix of some trace-equivalent 
word. This is done using several buffers, i.e., extending
the idea of buffered simulation from
\cite{DBLP:journals/corr/HutagalungLL14}: instead of using only one
buffer, there are several buffers of certain capacities and associated
(not necessarily disjoint) alphabets. Whenever \spoiler chooses a
letter, it is written to all those buffers whose alphabet contains
that letter. Dually, \duplicator can only use those letters that are
available at all the associated buffers. \duplicator can only win if
she does not leave any letter in any of the buffers for ever. With
this setup of game and winning condition, \duplicator effectively
attempts to produce a trace-equivalent word. The second part of the
winning condition is standard: if \spoiler produces an accepting run,
then \duplicator's run has to be accepting as well.

Our main results in this context are the following:
\begin{itemize}
\item If \duplicator has a winning strategy, then the language of the
  first automaton is contained in the trace closure of the language of
  the second automaton (Thm.~\ref{thm:approximate}). While the latter
  property is undecidable, the existence of a winning strategy with
  buffers of finite capacities is decidable in polynomial time
  (provided the number and capacities of buffers are unchanged,
  Thm.~\ref{thm:decidability}).
\item From \cite{DBLP:journals/corr/HutagalungLL14}, we know that
  buffered simulation (using a single unbounded buffer) is
  decidable. Section~\ref{sec:undecidability} proves that adding a
  single bounded buffer yields a highly undecidable simulation
  relation, as hard as recursive B\"uchi games and therefore hard for
  the class of all Boolean combinations of $\Sigma^1_1$-problems
  (Thm.~\ref{thm:RBG2simu}).
\item Section~\ref{sec:sufficiency} describes the simulation relations
  in terms of continuous functions between the accepting runs of the
  two automata. This yields completeness results in the sense that
  multi-buffer simulation implies trace-closure inclusion in certain
  cases.
\end{itemize}

\section{B\"uchi Automata and Trace Languages}
\label{sec:prel}

Let $\Sigma$ be an alphabet. Then $\Sigma^*$ denotes the set of finite
words over $\Sigma$, $\Sigma^\omega$ is the set of all infinite words
over $\Sigma$, and $\Sigma^\infty=\Sigma^*\cup\Sigma^\omega$. For a
natural number $k$, we set $[k]=\{1,2,\dots,k\}$.

A \emph{nondeterministic B\"uchi automaton} or \emph{NBA} is a tuple
$\mathcal A=(Q,\Sigma,q_{\mathsf{I}},\delta,F)$ where $Q$ is a finite
set of \emph{states}, $\Sigma$ is an alphabet, $q_{\mathsf{I}}\in Q$
is the \emph{initial state},
$\delta\colon Q\times\Sigma\to\mathcal{P}(Q)$ is the \emph{transition
  function}, and $F\subseteq Q$ is the set of \emph{accepting states}.

Let $w=a_0a_1a_2\dots\in\Sigma^\omega$ be an infinite word over
$\Sigma$. A \emph{run} of $\mathcal A$ on $w$ is an alternating
sequence of states and letters $\rho=(q_0,a_0,q_1,a_1,\dots)$ with
$q_0=q_{\mathsf{I}}$ and $q_{i+1}\in\delta(q_i,a_i)$ for all
$i\ge0$. This run is \emph{accepting} if $q_i\in F$ for infinitely
many $i\in\Nat$. The \emph{language $L(\mathcal A)$ of $\mathcal A$}
is the set of infinite words that admit an accepting run.

The main motivation of this paper is to approximate inclusion of trace
languages. Therefore we shortly introduce the notions of finite and
infinite traces, for a detailed treatment see \cite{DieR95}.

A \emph{trace alphabet} is a tuple $\sigma=(\Sigma_i)_{i\in[k]}$ of
not necessarily disjoint alphabets (note that $k$ is arbitrary). Let
$\Sigma=\bigcup_{i\in[k]}\Sigma_i$ and, for $a\in\Sigma$, let
$\sigma(a)=\{i\in[k]\mid a\in\Sigma_i\}$ which is by construction
nonempty. The idea is that the letter $a\in\Sigma$ denotes an action
that is performed by the set of processes $\sigma(a)$. For $i\in[k]$,
$\pi_i\colon\Sigma^\infty\to\Sigma_i^\infty$ is the natural projection
function that deletes from each word all letters that do not belong
to $\Sigma_i$.  We call two words $u,v\in\Sigma^\infty$
\emph{$\sigma$-equivalent} if $\pi_i(u)=\pi_i(v)$ for all
$i\in[k]$. In this case, we write $u\sim_\sigma v$. The relation
$\sim_\sigma$ is called \emph{trace equivalence}.

The restriction of $\sim_\sigma$ to $\Sigma^*$ has an alternative
characterisation (that is the traditional definition of trace
equivalence): let
$D=\bigcup_{i\in[k]}\Sigma_i\times\Sigma_i\subseteq\Sigma^2$ denote
the set of pairs $(a,b)$ with
$\sigma(a)\cap\sigma(b)\neq\emptyset$. This reflexive and symmetric
relation is called the \emph{dependence relation} associated
with~$\sigma$. Then the restriction of $\sim_\sigma$ to $\Sigma^*$ is
the least congruence on the free monoid $\Sigma^*$ with
$ab\sim_\sigma ba$ for all $(a,b)\notin D$.\footnote{Given a reflexive
  and symmetric relation $D\subseteq\Sigma^2$, one can always find a
  tuple $(\Sigma_i)_{i\in[k]}$ that induces $D$ (where $k$ depends on
  $D$.)}
The quotient $\mathbb{M}(\sigma)=\Sigma^*/_{\sim_\sigma}$ is
called \emph{trace monoid}, its elements are \emph{finite traces}.
The quotient $\mathbb{R}(\sigma)=\Sigma^\omega/_{\sim_\sigma}$ is
the set of \emph{real} or \emph{infinite traces}.
The \emph{trace closure} of a language $L\subseteq\Sigma^\infty$ w.r.t.\
$\sigma$ is the language
$[L]_\sigma=\{v\in\Sigma^\infty\mid \exists u\in L\colon
u\sim_\sigma v\}$. The language $L$ is \emph{trace closed} if it
equals its trace closure.

\begin{example}
  Let $\Sigma=\{a,b,c\}$ with $\sigma(a)=\{1\}$, $\sigma(b)=\{1,2\}$,
  and $\sigma(c)=\{2\}$. Then $a^*(bc)^*$ is trace closed, the trace
  closure of $a^*c^*$ is the language $\{a,c\}^*$ and the trace
  closure of $(ac)^*$ is the language of all words $u\in\{a,c\}^*$
  with the same numbers of occurrences of $a$ and $c$, resp.
\end{example}

Let the mapping $\sigma$ be such that the induced independence
relation $\Sigma^2\setminus D$ is not transitive. Then, given a
regular language $L\subseteq\Sigma^*$, it is undecidable whether its
trace closure $[L]_\sigma$ is regular~\cite{Sak92}. Even more, it is
undecidable whether the closure is universal, i.e., equals
$\Sigma^*$. Consequently, for two regular languages $K$ and $L$, it is
undecidable whether $K\subseteq[L]_\sigma$ (which is equivalent to
$[K]_\sigma\subseteq[L]_\sigma$). These negative results also hold for
languages of infinite words and their trace closures.

\section{Multi-Buffer Simulations}
\label{sec:multibuf}

Let $\sigma=(\Sigma_i)_{i\in[k]}$ be a trace alphabet and
$\mathcal{A} =
(Q^{\mathcal{A}},\Sigma,p_I,\delta^{\mathcal{A}},F^{\mathcal{A}})$ and
$\mathcal{B} =
(Q^{\mathcal{B}},\Sigma,q_I,\delta^{\mathcal{B}},F^{\mathcal{B}})$ be
two automata over the alphabet $\Sigma$. We aim at finding an
approximation to the undecidable question of
$L(\mathcal A)\subseteq[L(\mathcal B)]_\sigma$ via simulation
relations. In these, we have $k$ FIFO buffers and
$\sigma(a)\subseteq[k]$ is interpreted as the set of buffers that are
used to transmit the letter $a$.
Let $\kappa\colon [k] \to \mathbb{N} \cup \{ \omega \}$ be a function
that assigns a \emph{capacity} to each buffer, i.e.\ the maximum
number of letters that this buffer can contain at any time. We will
often write such a function as a tuple $(\kappa(1),\ldots,\kappa(k))$.

The multi-buffer game
$\DKsimugame{\kappa}{\sigma}{\mathcal{A}}{\mathcal{B}}$ is played on
these two automata and the $k$ buffers between players \spoiler and
\duplicator as follows. Configurations are tuples
$(p,\beta_1,\beta_2,\dots,\beta_k,q)\in Q^{\mathcal{A}} \times
\Sigma_1^* \times \ldots \times \Sigma_k^* \times Q^{\mathcal{B}}$
with $|\beta_i|\le\kappa(i)$ for all $i\in[k]$. The first and last
component can be seen as the places of two tokens on the state spaces
of $\mathcal{A}$ and $\mathcal{B}$ respectively; the others denote the
current buffer contents. The initial configuration is
$(p_I,\epsilon,\ldots,\epsilon,q_I)$.  A round consists of a move by
\spoiler followed by a move by \duplicator. \spoiler choses
$a\in\Sigma$, moves the token in $\mathcal{A}$ forward along an
$a$-transition of his choice, and pushes a copy of the $a$-symbol
to each of the buffers from $\sigma(a)$. Then \duplicator either skips
her turn or chooses a non-empty word $a_1\dots a_n\in \Sigma^+$ and
moves the token in $\mathcal{B}$ along some $a_1\dots a_n$-labeled
path. While doing so, for every $i$, she pops an $a_i$ from each of
the buffers from $\sigma(a_i)$. More formally, in a configuration of
the form $(p,\beta_1,\ldots,\beta_k,q)$,
\begin{enumerate}
\item \spoiler picks a letter $a \in \Sigma$ and a state
  $p' \in Q^\mathcal{A}$ such that $\Transition{p}{a}{p'}$, and
  outputs $ap'$.
\item \duplicator picks a finite run $qv_1q_1v_2q_2 \cdots v_{n}q'$
  from $q$ in the automaton $\mathcal {B}$ such that
  $\pi_i(a) \beta_i= {\beta_i}'\pi_i(v_1v_2\dots v_n)$ for all
  $b \in [k]$.  She outputs $v_1q_1v_2q_2 \cdots v_{n}q'$. 
\end{enumerate}
The play proceeds in the configuration
$(p',{\beta_1}',\ldots,{\beta_k}',q')$.

Since $(p',{\beta_1}',\ldots,{\beta_k}',q')$ is a configuration, we
implicitely have $|{\beta_i}'|\le\kappa(i)$, i.e., the size of the
buffers is checked \emph{after} the round. So \spoiler can write into
a ``full'' buffer (i.e., with $|\beta_i|=\kappa(i)$) and it is
\duplicator's responsibility to shorten the buffer again. In
particular, \duplicator has to read all letters from buffers with
capacity 0 in the very same round. Furthermore, if \spoiler uses
buffers of finite capacity infinitely often, then \duplicator cannot
skip forever.

A finite play is \emph{lost} by the player that got stuck (which, for
\spoiler, means that he gets trapped in a sink of $\mathcal{A}$ while,
for \duplicator, it means that she should shorten a buffer but cannot
do so).  An infinite play produces an infinite run
$\rho_{\mathcal{A}}$ of $\mathcal{A}$ over some infinite word
$w_{\mathcal{A}}\in\Sigma^\omega$, and a finite or infinite run
$\rho_{\mathcal{B}}$ of $\mathcal{B}$ over some word
$w_{\mathcal{B}}\in\Sigma^\infty$. This play is \emph{won} by
\duplicator iff
\begin{itemize}
\item $\rho_{\mathcal{A}}$ is not an accepting run, or
\item $\rho_{\mathcal{B}}$ is an infinite accepting run and
  every letter written by \spoiler into a buffer will eventually
  be read by \duplicator (formally: for every letter $a\in\Sigma$, the
  numbers of occurrences of $a$ in $w_{\mathcal A}$ and in
  $w_{\mathcal B}$ are the same).
\end{itemize}
We write $\mathcal{A} \DKsimu{\kappa}{\sigma} \mathcal{B}$ if
\duplicator has a winning strategy for the game
$\DKsimugame{\kappa}{\sigma}{\mathcal{A}}{\mathcal{B}}$.

\begin{example}
  \label{ex:multibuff}
  Consider the trace alphabet $\sigma$ with $\Sigma_1=\{a,b\}$,
  $\Sigma_2=\{b\}$ and $\Sigma_3=\{c\}$ and the following two NBA
  $\mathcal{A}$ (top) and $\mathcal{B}$ (below) over the alphabet
  $\Sigma$.  \vspace*{-4mm}

\noindent
\begin{minipage}{0.575\textwidth} \setlength{\parfillskip}{0pt}
  \quad\enspace We have
  $\mathcal{A} \DKsimu{(\omega,2,0)}{\sigma} \mathcal{B}$.  Note that
  in this game, $a$ and $b$ get put into an unbounded buffer, $b$ also
  gets put into a buffer of capacity $2$, and $c$ gets put into a
  buffer of capacity $0$, i.e.\ \duplicator has to respond immediately
  to any $c$-
\end{minipage}
\raisebox{2mm}{\begin{minipage}{0.41\textwidth}
\begin{tikzpicture}[initial text={}, node distance=12mm, every state/.style={minimum size=3mm,inner sep=0pt}]
  \node[state,initial]   (qa0)                      {};
  \node[state]           (qa1) [right of=qa0]       {};
  \node[state]           (qa2) [right of=qa1]       {};
  \node[state,accepting] (qa3) [right of=qa2]       {};
  
  \path[->] (qa0) edge              node [above]     {$b$}        (qa1)
            (qa1) edge              node [above,near start]      {$b$}        (qa2)
            (qa2) edge [loop above] node [left]      {$a$}        (qa2
            (qa2) edge [bend left] node [above]      {$c$}        (qa3)
            (qa3) edge [bend left] node [above]      {$a$}        (qa2);

  \node[state,initial]   (qb0) [below=7mm of qa0] {};
  \node[state]           (qb1) [right of=qb0]                      {};
  \node[state]           (qb2) [right of=qb1]                      {};
  \node[state,accepting] (qb3) [right of=qb2]                      {};
  \node[state]           (qb4) [right of=qb3]                      {};
  
  \path[->] (qb0) edge              node [above]      {$c$}        (qb1)
            (qb1) edge              node [above]      {$b$}        (qb2)
            (qb2) edge              node [above]      {$b$}        (qb3)
            (qb3) edge [bend left]  node [above]      {$c$}        (qb4)
            (qb4) edge [bend left]  node [below]      {$a$}        (qb3);

\end{tikzpicture}
\end{minipage}} \\
move made by \spoiler. \duplicator's winning strategy consists of
skipping her turn until \spoiler produces a $c$. Note that he cannot
produce more than 2 $b$'s beforehand, hence he cannot win by exceeding
the capacity of the second buffer. Note also that he cannot loop on
the first $a$-loop for ever, otherwise he will lose for not producing
an accepting run. Once \spoiler eventually produced a~$c$, \duplicator
consumes it together with the entire content of the second buffer and
moves to the accepting state in her automaton. After that she can
immediately respond to every state-changing move by \spoiler.
\end{example}

The following theorem shows indeed that multi-buffer games approximate
the inclusion between the trace closures of the languages of two NBA.

\begin{theorem}
  \label{thm:approximate}
  Let $\sigma=(\Sigma_i)_{i\in[k]}$ be a trace alphabet and let
  $\kappa$ be a capacity function for $k$ buffers.  Let $\mathcal{A}$
  and $\mathcal{B}$ be two NBA over $\Sigma$ with
  $\mathcal{A} \DKsimu{\kappa}{\sigma} \mathcal{B}$. Then
  $L(\mathcal{A}) \subseteq [L(\mathcal{B})]_\sigma$.
\end{theorem}

\begin{proof}
  Let $w_{\mathcal A}=a_0a_1a_2\dots\in L(\mathcal A)$ be
  arbitrary. Then \spoiler can play such that $\rho_{\mathcal A}$ is
  an accepting run over $w_{\mathcal A}$. Since \duplicator has a
  winning strategy, she can play in such a way that also
  $\rho_{\mathcal B}$ is an accepting run and no letter remains in a
  buffer for ever. Now let $1\le i\le k$. Then
  $\pi_i(w_{\mathcal A})\in\Sigma_i^\infty$ is the sequence of letters
  that \spoiler writes into the buffer $i$ during the play. Since
  \duplicator can only execute letters that are available at the
  corresponding buffers, the word $\pi_i(w_{\mathcal B})$ is a prefix
  of~$\pi_i(w_{\mathcal A})$. If it is a proper prefix, then
  \duplicator failed to read all letters written into buffer $i$. As
  \duplicator plays according to her winning strategy, this is not the
  case. Hence $\pi_i(w_{\mathcal A})=\pi_i(w_{\mathcal B})$. Since
  this holds for all $i\in[k]$, we have
  $w_{\mathcal A}\sim_\sigma w_{\mathcal B}$ and therefore
  $L(\mathcal A)\subseteq[L(\mathcal B)]_\sigma$. 
\end{proof}

This yields, together with the following observation, a sound (but not
necessarily complete) approximation procedure for trace language
inclusion problems using bounded buffers.

\begin{theorem}
  \label{thm:decidability}
  Uniformly in the trace alphabet $\sigma=(\Sigma_i)_{i\in[k]}$ and
  the capacity function $\kappa\colon[k]\to\Nat$, the relation
  $\DKsimu{\kappa}{\sigma}$ is decidable on automata with $m$ and $n$
  states, resp., in time $O((k+1)\cdot(mn|\Sigma|^{r+k}(k+1))^{2.5})$
  where $r = \kappa(1) + \ldots + \kappa(k)$.
\end{theorem}

If we fix $k$ and the capacity function $\kappa$, then this time bound
reduces to a polynomial in $mn|\Sigma|$, i.e., in the size of the
automata $\mathcal A$ and $\mathcal B$.

\begin{proof}
  Let $\mathcal A$ and $\mathcal B$ be automata with $m$ and $n$
  states, resp. Then
  $\DKsimugame{\kappa}{\sigma}{\mathcal{A}}{\mathcal{B}}$ can be
  understood as a game whose positions consist of a configuration and
  an element of $\{0,1,\dots,k\}$ to store one of the buffers
  \duplicator used in her last move ($0$ stands for ``\duplicator
  skiped her move''), i.e., of finite size
  $ \le m\cdot n\cdot \prod_{i=1}^k |\Sigma_i|^{\kappa(i)+1}\cdot
  (k+1) \le mn \cdot|\Sigma|^{r+k}\cdot(k+1)$. Its winning condition
  is a strong fairness condition (``\emph{if} \spoiler visits final
  states infinitely often \emph{then} so does \duplicator'') together
  with $k$ B\"uchi-conditions (``infinitely often, \duplicator reads
  some letter from buffer $i$ or buffer $i$ is empty''). By
  \cite{Hen16}, such games can be solved in time
  $O((k+1)\cdot(mn|\Sigma|^{r+k}(k+1))^{2.5})$.
\end{proof}

Multi-buffer simulations form a hierarchy in the sense that
\duplicator's power strictly grows with the buffer capacities.

\begin{theorem}
  \label{thm:hierarchies}
  Let $\sigma=(\Sigma_i)_{i\in[k]}$ be a trace alphabet and let
  $\kappa,\kappa'$ be capacity functions for $k$ buffers.

  If $\kappa(i) \le \kappa'(i)$ for all $i \in [k]$, then
  $\DKsimu{\kappa}{\sigma} \subseteq \DKsimu{\kappa'}{\sigma}$.

  Moreover, if there are $a\in\Sigma_i$ and $b\in\Sigma$ with
  $\sigma(a)\cap\sigma(b)=\emptyset$ and $\kappa(i)<\kappa'(i)$, then
  there are automata $\mathcal{A}$ and $\mathcal{B}$ such that
  $\mathcal{A} \DKnotsimu{\kappa}{\sigma} \mathcal{B}$ but
  $\mathcal{A} \DKsimu{\kappa'}{\sigma} \mathcal{B}$.
\end{theorem}

\begin{proof}
  We immediately get
  $\DKsimu{\kappa}{\sigma} \subseteq \DKsimu{\kappa'}{\sigma}$ for
  $\kappa \le \kappa'$ since any winning strategy for \duplicator in
  $\DKsimugame{\kappa}{\sigma}{\mathcal{A}}{\mathcal{B}}$ is also a
  winning strategy for her in
  $\DKsimugame{\kappa'}{\sigma}{\mathcal{A}}{\mathcal{B}}$.
  \vspace*{-3.4mm}

\noindent
\parbox{0.69\textwidth}{ \quad\enspace For the strictness part suppose
  w.l.o.g.\ $a\in\Sigma_1$, $b\in\Sigma_2$,
  $\sigma(a)\cap\sigma(b)=\emptyset$, and $\kappa(1)<\kappa'(1)$. Then
  consider these two NBA $\mathcal{A}$ (top) and $\mathcal{B}$ (below)
  over $\Sigma$. \duplicator wins the game
  $\DKsimugame{\kappa'}{\sigma}{\mathcal A}{\mathcal B}$ by simply
  choosing $ba^{\kappa(1)+1}$ every $\kappa(1)+1$ rounds (and skipping
  in the other rounds). \spoiler wins the game
  $\DKsimugame{\kappa}{\sigma}{\mathcal A}{\mathcal B}$ choosing $a$
  in the first $\kappa(1)+1$ rounds such that \duplicator is forced to
  skip in the first}
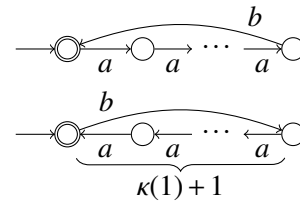
\raisebox{3mm}{\begin{minipage}{0.29\textwidth}
\begin{tikzpicture}[initial text={}, node distance=10mm, every state/.style={minimum size=3mm,inner sep=0pt}]
  \node[state,initial,accepting]   (qa0)                      {};
  \node[state]                     (qa1) [right of=qa0]       {};  
  \node        		           (dot) [right of=qa1] 	    {$\cdots$};
  \node[state]                     (qa2) [right of=dot]       {};

  \path[->] (qa0) edge                 node [below]      {$a$}                       (qa1)
	    (qa1) edge 		       node [below]      {$a$} 	                    (dot)
	    (dot) edge                 node [below]      {$a$}                       (qa2)
	    (qa2) edge [bend right=20] node [above,very near start]      {$b$}     (qa0);
          
  \node[state,initial,accepting]   (qb0)  [below=8mm of qa0]   {};
  \node[state]                     (qb1)  [right of=qb0]                      {};  
  \node        		           (dot2) [right of=qb1]                      {$\cdots$};
  \node[state]                     (qb2)  [right of=dot2]                     {};

  \path[->] (qb1) edge                node [below]                        {$a$}     (qb0)
	    (dot2) edge               node [below]                        {$a$}     (qb1)
	    (qb2) edge                node [below]      {$a$}     (dot2)
	    (qb0) edge [bend left=20] node [above,very near start]      {$b$}     (qb2);
          
\draw [decorate,decoration={brace,amplitude=5pt,mirror,raise=2mm}] 
(qb0.south east) -- (qb2.south west) node [midway, below=3mm] {$\kappa(1)+1$};
\end{tikzpicture}
\end{minipage}} \\[.8mm]
$\kappa(1)+1$ rounds (since no $b$ is available in the second
buffer) which exceeds the capacity of the second buffer.
\end{proof}

Thms.~\ref{thm:approximate}, \ref{thm:decidability} and
\ref{thm:hierarchies} can be used for an incremental inclusion test:
suppose we want to check whether
$L(\mathcal A) \subseteq [L(\mathcal B)]_\sigma$ holds for the
trace alphabet $\sigma=(\Sigma_i)_{i\in[k]}$.  First consider
$\kappa_0$ with $\kappa_0(i)=0$ for all $i \in [k]$. If
$\mathcal{A} \DKsimu{\kappa_0}{\sigma} \mathcal{B}$, then
$L(\mathcal A)\subseteq L(\mathcal B)\subseteq [L(\mathcal
B)]_\sigma$. If this is not the case, chose $\kappa_1$ with
$\kappa_0(i)\le\kappa_1(i)$ for all $i$ and
$\kappa_0(i) < \kappa_1(i)$ for some $i$. If
$\mathcal{A} \DKsimu{\kappa_1}{\sigma} \mathcal{B}$, then
$L(\mathcal A)\subseteq [L(\mathcal B)]_\sigma$. If, again, this
fails, then extend the buffer capacities to some $\kappa_2$ etc.
Sect.~\ref{sec:sufficiency} analyses completeness of this procedure,
i.e.\ the possibility for this to prove trace language non-inclusion.

\section{Undecidability}
\label{sec:undecidability}

It is not hard to show that multi-buffer simulation is in general
undecidable by a reduction from Post's Correspondence Problem (PCP)
adapting the argument used for the reachability problem for
communicating finite state machines~\cite{BrandZ83,ChambartS08} and
yielding $\Pi_1^0$-hardness of
$\DKsimu{(\omega,\omega)}{(\{a,b\},\{c,d\})}$.  There is a variant of
PCP called $\omega$PCP(REG) that is known to be $\Sigma^1_1$-complete
\cite{journals/ijfcs/Finkel12}. It asks for the existence of an
infinite solution word that additionally belongs to some
$\omega$-regular language.  It is not difficult to adjust the
reduction to the $\DKsimu{(\omega,\omega)}{(\{a,b\},\{c,d\})}$-problem
such that \spoiler's accepting runs correspond to valid
solutions. This would yield $\Pi^1_1$-hardness of
$\DKsimu{(\omega,\omega)}{(\{a,b\},\{c,d\})}$.

We do not give details of this reduction here because it is still
possible to strengthen the undecidability result in two ways: (1) we
will show that a single unbounded buffer suffices for
undecidability. Note that $\DKsimu{(\omega)}{\Sigma}$ is decidable in
EXPTIME \cite{DBLP:journals/corr/HutagalungLL14}.  Hence, the question
is how many additional bounded buffers are needed to establish
undecidability. We provide a tight result in this respect showing that
$\DKsimu{(\omega,0)}{(\{a,b\},\{c\})}$ is undecidable already, i.e.\
the addition of a minimal number of buffers of minimal capacity
(contrary to this, $\DKsimu{(\omega,0)}{(\{a\},\{c,d\})}$ is shown to
be decidable). (2) We will show that the level of undecidability is
genuinely higher than $\Pi^1_1$ by considering the problem of solving
a B\"uchi game played on a recursive game graph.

A \emph{recursive B\"uchi game} (RBG) is a graph
$G = (V,E,\mathit{Own},\mathit{Fin},v_{\mathsf{I}})$ such that $V$ is
a decidable set of nodes, $\mathit{Own}$ and $\mathit{Fin}$ are
decidable subsets of $V$, $v_{\mathsf{I}}$ is a designated starting
node, and $E \subseteq V\times V$ is a decidable set of edges. The
game is played between players $0$ and $1$ starting in
$v_0 := v_{\mathsf{I}}$. Whenever it reaches a node $v_i$ and
$v_i \in \mathit{Own}$ then player $0$ chooses $v_{i+1}\in V$ such
that $(v_i,v_{i+1}) \in E$ and the play continues with
$v_{i+1}$. Otherwise player $1$ chooses such a node $v_{i+1}$.

A player wins a play if the opponent is unable to choose a successor
node. Moreover, player $0$ wins an infinite play $v_0, v_1, \ldots$ if
there are infinitely many $i$ such that $v_i \in \mathit{Fin}$. The
\emph{recursive B\"uchi game problem} is to decide, given such a game
represented using Turing machines, whether or not player $0$ has a
winning strategy for this game.

The existence of a winning strategy for player 0 in an RBG is a
typical $\Sigma^1_2$-statement (``\emph{there exists} a strategy for
player 0 \emph{such that all} plays $(v_i)_{i\ge0}$ conforming to this
strategy satisfy
$\forall n\in\Nat\,\exists m\in\Nat\colon v_{m+n}\in\mathit{Fin}$''),
i.e., the RBG problem belongs to $\Sigma^1_2$. By determinacy of Borel
(and therefore of B\"uchi) games~\cite{Mart75}, the existence of a
winning strategy for player 0 is equivalent to the non-existence of a
winning strategy for player 1 (i.e., to ``\emph{for all} strategies of
player 1 \emph{there exists} a play conforming to this strategy
satisfying
$\forall n\in\Nat\,\exists m\in\Nat\colon
v_{m+n}\in\mathit{Fin}$''). Hence the RBG problem also belongs to
$\Pi^1_2$ and therefore to $\Sigma^1_2\cap\Pi^1_2$. This class does
not contain any complete problems \cite[Thm.~16.1.X]{Rog67}, but we
can show the following lower bound for the RBG problem.

\begin{theorem}
\label{thm:RGBhard}
  The recursive B\"uchi game problem is hard for the class
  $B\Sigma^1_1$ of all Boolean combinations of problems from
  $\Sigma^1_1$.
\end{theorem}

\begin{proof}
  To see this, recall that the set of (pairs of Turing machines
  accepting the nodes and edges of) recursive trees with an infinite
  branch is $\Sigma^1_1$-hard \cite{Kle43}. It follows that the class
  of tuples $(S_i,T_i)_{1\le i\le n}$ of recursive trees such that,
  for some $1\le i\le n$, the tree $S_i$ has an infinite branch while 
  $T_i$ does not, is complete for the class $B\Sigma^1_1$. We
  reduce this problem to the recursive B\"uchi game problem. So let
  $(S_i,T_i)_{1\le i\le n}$ be a tuple of recursive trees. We build a
  B\"uchi game as follows:
 First, to any tree $S_i$, we add a node $g_i$ together with
    edges from all nodes (including $g_i$ itself) to~$g_i$. The set
    $\mathit{Own}$ equals the set of nodes of $S_i$ plus this
    additional node $g_i$. Next, we replace every edge by a path of
    length $2$ (i.e., with two edges). The set $\mathit{Fin}$ of
    winning nodes are the original nodes from $S_i$. Starting in
    $v_0$, player 0 has a winning strategy of this B\"uchi game $G_i$
    iff the tree $S_i$ contains an infinite path.
    
  Similarly, to any tree $T_i$, we add a node $h_i$ together
    with edges from all nodes (including $h_i$ itself) to $h_i$. Next,
    we replace every edge by a path of length $2$. The set
    $\mathit{Own}$ consists of these new nodes. The node $h_i$ and the
    unique successor node $h_i'$ (that originates from the replacement
    of the edge $(h_i,h_i)$ by a path of length 2) are the only
    winning nodes from $\mathit{Fin}$. Starting in the root of $T_i$,
    player 0 has a winning strategy in this B\"uchi game $H_i$ iff
    $T_i$ does not contain any infinite path.
    Note that once a play enters a winning node it will continue with
    winning nodes \textit{ad infinitum}.
    
 For any $i$ with $1\le i\le n$, we construct the direct
    product of the two games $G_i$ and $H_i$ described above: Nodes
    are of the form $(g,h)$ where $g$ is a node from $G_i$ and $h$ a
    node from $H_i$ with
    $g\in\mathit{Own}_{G_i}\iff h\in\mathit{Own}_{G_i}$. The set
    $\mathit{Own}$ equals
    $\mathit{Own}_{G_i}\times\mathit{Own}_{H_i}$. There is an edge
    from $(g,h)$ to $(g',h')$ iff there are edges from $g$ to $g'$ in
    $G_i$ and from $h$ to $h'$ in $H_i$. Finally, a node $(g,h)$
    belongs to $\mathit{Fin}$ iff both $g$ and $h$ are winning.
    Clearly, any play in this game $GH_i$ ``consists'' of two plays in
    $G_i$ and in $H_i$, resp. Since a play in $S_i$ cannot leave the
    set of winning nodes, a play in this game is won by player 0 if
    both component plays are won by player 0 in $G_i$ and in $H_i$,
    resp. Hence player 0 has a winning strategy iff $S_i$ contains an
    infinite branch while $T_i$ does not.
    
 Finally, we consider the disjoint union of all the games
    $GH_i$ and add a node $v_{\mathsf{I}}\in\mathit{Own}$. In
    addition, we add edges from $v_{\mathsf{I}}$ to the starting nodes
    of all the games $GH_i$. Now it is rather obvious that player 0
    wins this game iff it wins one of the games $GH_i$ and therefore
    iff, for some $1\le i\le n$, the tree $S_i$ contains an infinite
    branch and $T_i$ does not.

  From Turing machines that describe the trees $S_i$ and $T_i$, we can
  construct Turing machines that describe this game. Hence, we reduced
  a $B\Sigma^1_1$-complete problem to the RBG problem.
\end{proof}

The rest of this section is devoted to showing that
$\DKsimu{(\omega,0)}{(\{a,b\},\{c,d\})}$ is computationally at least
as difficult as solving general RBGs. We present a reduction from the
RBG problem to $\DKsimu{(\omega,0)}{(\{a,b\},\{c,d\})}$. Let $G$ be a
RBG. Using standard encoding tricks we can assume that its node set is
$\{0,1\}^+$, $\mathit{Own}=0\{0,1\}^*$, the initial node is $1$, and
$\mathit{Fin}=\{0,1\}^*1$. The edge relation of $G$ is decided by a
deterministic Turing Machine~$\mathcal{M}$ with state set $Q$, tape
alphabet $\Gamma$, and transition function
$\delta\colon Q \times \Gamma \to Q \times \Gamma \times \{-1,0,1\}$.
W.l.o.g., we can assume that $\mathcal{M}$ has designated initial /
accepting / rejecting states $\mathsf{init}$ / $\mathsf{acc}$ /
$\mathsf{rej}$ and that the tape alphabet $\Gamma$ equals
$\{0,1,\#,\triangleright,\triangleleft\}$ including a division symbol
$\#$ and two end-of-tape markers $\triangleleft$
and~$\triangleright$. Apart from the usual assumption that
$\mathcal{M}$ uses the end-of-tape markers sensibly, we presume the
following.
\begin{itemize}
\item There are two designated states $\mathsf{acc}$ and
  $\mathsf{rej}$ that the machine uses to signal acceptance and
  rejection.
\item When started in the configuration
  $\triangleright\, w\,\#\,v\, \mathsf{init}\,\triangleleft$ with
  $v,w \in \{0,1\}^+$, the machine eventually halts in
  $\triangleright\, w\, \mathsf{acc}\,\triangleleft$ if $w$ is a
  successor of $v$; otherwise it halts in
  $\triangleright\, w\, \mathsf{rej}\,\triangleleft$. Thus, we assume
  it to reproduce the name of the node that it checked for being a
  successor node to $v$. This helps a subsequent computation to be set
  up. Also note that we assume the machine's tape to be infinite to
  the left and that it starts reading its input from the right. This
  is purely done for presentational purposes since it better matches
  the use of buffers in the constructed multi-buffer games.
\end{itemize} 
In order to ease the presentation we derive a function
$\hat{\delta}: (\Gamma \cup Q)^4 \to (\Gamma \cup Q)^{\le5}$ from the
transition function $\delta$ such that, for any configuration
$\triangleright a_1 a_2 \dots a_k \triangleleft$, the unique successor
configuration equals
\[
  \hat\delta(\triangleright,a_1,a_2,a_3)\,
  \hat\delta(a_1,a_2,a_3,a_4)\,
  \hat\delta(a_2,a_3,a_4,a_5)\,
  \dots
  \hat\delta(a_{k-2},a_{k-1},a_k,\triangleleft)\,.
\]
If $\triangleright,\triangleleft\notin\{b_1,b_2,b_3,b_4\}$, then we
have $|\hat\delta(b_1,b_2,b_3,b_4)|=1$,
$|\hat\delta(\triangleright,b_1,b_2,b_3)|\in\{2,3\}$,
$|\hat\delta(b_1,b_2,b_3,\triangleleft)|\in\{3,4\}$, and
$|\hat\delta(\triangleright,b_1,b_2,\triangleleft)|\in\{4,5\}$.%
 
The construction of two automata $\mathcal{A}$ and $\mathcal{B}$ from
the RBG $G$ hinges on a simple correspondence between winning
strategies in $G$ and those in
$\DKsimugame{(\omega,0)}{\sigma}{\mathcal{A}}{\mathcal{B}}$: \spoiler and
\duplicator simulate the RBG by players $1$ and $0$, respectively.
The $\omega$-buffer is used to name current nodes of the RBG. Its
alphabet contains all symbols used to form configurations:
$\Sigma_1 := \Gamma \cup Q$. The alphabet of letters that can be put
into the capacity-$0$ buffer contains a special new symbol and a copy
of every $\Sigma_1$-symbol:
$\Sigma_2 := \{c\} \cup \{c_x \mid x \in \Sigma_1\}$. 

We need to show how three aspects of the simulation can be realised:
\begin{enumerate}
\item The choice of a successor node by player $1$ in $G$. This is
  easy since player $1$ is simulated by \spoiler. It is easy to
  construct a $\mathcal{A}_{\mathsf{chs}}$ that allows \spoiler to
  choose a $v \in \{0,1\}^+$ and put $\triangleright\, v$ into
  \vskip1mm

  \noindent
  \begin{minipage}{0.67\textwidth} \setlength{\parfillskip}{0pt} the
    buffer. The automaton is shown on the right.  It has two states
    marked with incoming and outgoing edges. These are used to
    indicate in which state a play should begin and where it should
    end
    \end{minipage}\quad  
    \raisebox{0mm}{\begin{minipage}{0.3\textwidth}
    \begin{tikzpicture}[every state/.style={minimum size=3mm,inner sep=0pt}, node distance=13mm, baseline=-1mm]
      \node[state]           (a0) {};
      \node[state]           (a1) [right of=a0] {};
      \node[state]           (a2) [right of=a1] {};
      \path[->] (a0) edge              node [above] {$0,1$}                               (a1) 
                (a1) edge [loop above] node [left]  {$0$}              node [right] {$1$} () 
                     edge              node [above] {$\triangleright$}                    (a2);
      \path[<-,draw] (a0) -- ++(-.5,0);
      \path[->,draw] (a2) -- ++(.5,0);
    \end{tikzpicture}
  \end{minipage}} \\[1mm]
according to the specification in these three cases. Later the
constructed automata will be plugged together by merging such marked
states forming NBA with infinite runs. We use final states at this
point only in special sinks that make one of the players win, i.e.\
make the opponent lose immediately.
\item The choice of a successor node by player $0$ in $G$. This is
  trickier, because we need to make \duplicator name a new node but it
  is only \spoiler who puts letters into the buffers. We will show in
  Lemma~\ref{lem:Dupforce} below how the capacity-$0$ buffer can be
  used in order for \duplicator to force \spoiler to produce a certain
  content for the $\omega$-buffer.
\item The check that a newly chosen node is indeed a successor of the
  current node. We make \spoiler produce a sequence of Turing machine
  configurations in the $\omega$-buffer and \duplicator check that
  they form an accepting computation. This is where the assumption of
  $\mathcal{M}$ being deterministic is needed because it forces
  \spoiler to produce an accepting computation if one exists.
  Lemma~\ref{lem:simucheck} below shows how this can be done.
\end{enumerate}

\begin{lemma}
  \label{lem:Dupforce}
  There are $\mathcal{A}_{\mathsf{frc}}$ and
  $\mathcal{B}_{\mathsf{frc}}$ with the following properties. Suppose
  the game
  $\DKsimugame{(\omega,0)}{(\Sigma_1,\Sigma_2)}{\mathcal{A}_{\mathsf{frc}}}{\mathcal{B}_{\mathsf{frc}}}$
  is played with the initial content of the $\omega$-buffer being
  $\#\,v\, \mathsf{init}\,\triangleleft$.  For every $w \in \{0,1\}^+$,
  \duplicator has a strategy in the game
  $\DKsimugame{(\omega,0)}{(\Sigma_1,\Sigma_2)}{\mathcal{A}_{\mathsf{frc}}}{\mathcal{B}_{\mathsf{frc}}}$
  to reach a configuration in which the buffer content is
  $\triangleright\, w\, \#\, v\, \mathsf{init}\, \triangleleft$.
\end{lemma}

\begin{figure}
\begin{center}
\begin{tikzpicture}[every state/.style={minimum size=3mm,inner sep=0pt}, node distance=13mm]
  \node[state]           (a0) {};
  \node[state]           (a1) [right of=a0] {};
  \node[state]           (a2) [right of=a1] {};
  \node[state]           (a3) [right of=a2] {};
  \node[state]           (a4) [below of=a1] {};
  \node[state]           (a5) [right of=a4] {};

 \path[->] (a0) edge                 node [below]                 {$c$}               (a1)
           (a1) edge                 node [below,near end]        {$c_{0}$}          (a2)
                edge [bend right=40] node [below,near end]        {$c_{1}$}          (a3)
                edge                 node [left]                  {$c_{\triangleright}$} (a4)
           (a2) edge [bend right]    node [above,very near start] {$0$}             (a0)
           (a3) edge [bend right]    node [above,very near start] {$1$}             (a0)
           (a4) edge                 node [above]                 {$\triangleright$} (a5);

  \path[<-,draw] (a0) -- ++(-.8,0);
  \path[->,draw] (a5) -- ++(.8,0);

  \node[state]           (b0) [right=3cm of a3] {};
  \node[state]           (b1) [right of=b0]     {};
  \node[state]           (b2) [right of=b1]     {};
  \node[state,accepting] (b3) [right of=b2]     {};
  \node[state]           (b4) [below of=b2]     {};
  \node[state]           (b5) [right of=b4]     {};
  \node[state]           (b6) [below of=b0]     {};

  \path[->] (b0) edge [bend right]    node [below,very near end]   {$c$}                         (b1)
                 edge [bend right=40] node [below,near start]      {$c$}                         (b2)
                 edge [bend right=30] node [left,very near start]  {$c$}                         (b4)
            (b1) edge [bend right]    node [below]                 {$c_{0}$}                    (b0)
                 edge [bend left]     node [above,near end]        {$\overline{c_{0}}$}         (b3)
            (b2) edge [bend right]    node [above]                 {$c_{1}$}                    (b0)
                 edge [bend left]     node [below,near start]      {$\overline{c_{1}}$}         (b3)
            (b3) edge [loop right]    node [right]                 {$\Sigma_1,\Sigma_2$}          ()
            (b4) edge                 node [above,very near start] {$\overline{c_{\triangleright}}$} (b3)
                 edge                 node [above,near end]        {$c_{\triangleright}$}            (b5)
            (b6) edge [bend right]    node [above,very near start] {$c$}                         (b1)
                 edge [bend right=20] node [below right=-1mm]           {$c$}                         (b2);

  \path[<-,draw] (b6) -- ++(-.8,0);
  \path[->,draw] (b5) -- ++(.8,0);

\end{tikzpicture}
\end{center}
\caption{Letting \duplicator force \spoiler to put something from $\triangleright\, \{0,1\}^+$ into the $\omega$-buffer.} 
\label{fig:Dupforce}
\end{figure}

\begin{proof}
  $\mathcal{A}_{\mathsf{frc}}$ is shown on the left of
  Fig.~\ref{fig:Dupforce}.  $\mathcal{B}_{\mathsf{frc}}$ is shown on
  the right using the abbreviations
  $\overline{c_a} = \Sigma_2 \setminus \{c_a\}$.

  Suppose the two players play on these automata starting in the
  states marked with incoming edges. 
  \spoiler must open the game by playing $c$, and \duplicator can
  respond to this synchronisation move going to a state that has
  exactly one outgoing edge that does not lead to the accepting state,
  labeled with either $c_0$ or $c_1$. \spoiler is now forced to play
  this $c_i$ for otherwise \duplicator will win by moving to the
  accepting state. In repsonse, \duplicator moves to the top left
  state and then \spoiler puts $i$ into the $\omega$-buffer. So,
  effectively, \duplicator has forced him to put $i \in \{0,1\}$ into
  the buffer with her choice in response to the $c$-move and the game
  proceeds with \spoiler in the initial state and \duplicator in the
  top left state. Note that here, the situation is similar, the only
  difference is that \duplicator now has the choice to go to three
  states as opposed to two before. If \duplicator (in repsonse to
  \spoiler's opening $c$-move) goes to his third option, then \spoiler
  is forced to put $\triangleright$ into the buffer and the play has
  reached the states marked with outgoing edges, and the content of
  the $\omega$-buffer is of the form
  $\triangleright\, w\, \#\, v\, \mathsf{init}\, \triangleleft$ with a
  $w \in \{0,1\}^+$ chosen by \duplicator, if it was
  $\#\, v\, \mathsf{init}\, \triangleleft$ at the beginning.
\end{proof}

\begin{lemma}
  \label{lem:simucheck}
  There are $\mathcal{A}_{\mathsf{chk}}$ and
  $\mathcal{B}_{\mathsf{chk}}$ with the following properties. Suppose
  the game
  $\DKsimugame{(\omega,0)}{(\Sigma_1,\Sigma_2)}{\mathcal{A}_{\mathsf{chk}}}{\mathcal{B}_{\mathsf{chk}}}$
  is played on these automata, and the content of the $\omega$-buffer
  is $C = \triangleright\, w\,\#\,v\,
  \mathsf{init}\,\triangleleft$. Then both players have a strategy to
  reach a configuration in which the buffer contains
  $\#\, w\, \mathsf{init}\,\triangleleft$ without losing in the
  meantime, iff $\mathcal{M}$ reaches an accepting configuration when
  started in $C$.
\end{lemma}

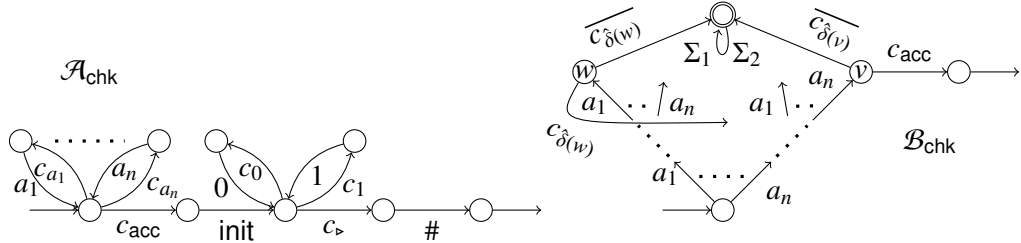
\begin{figure}[t]
\begin{center}
\begin{tikzpicture}[every state/.style={minimum size=3mm,inner sep=0pt}, node distance=13mm]
  \node[state]           (a0)                    {};
  \node[state]           (a1) [above left of=a0]  {};
  \node[state]           (a2) [above right of=a0] {};
  \node[state]           (a3) [right of=a0]       {};
  \node[state]           (a4) [right of=a3]       {};
  \node[state]           (a5) [above left of=a4]  {};
  \node[state]           (a6) [above right of=a4] {};
  \node[state]           (a7) [right of=a4]       {};
  \node[state]           (a8) [right of=a7]       {};

  \path[->] (a0) edge [bend right] node [below left=-2mm]  {$c_{a_1}$}        (a1)
                 edge [bend right] node [right=0mm]        {$c_{a_n}$}        (a2)
            (a1) edge [bend right] node [left=-1mm]        {$a_1$}           (a0)
            (a2) edge [bend right] node [below right=-2mm] {$a_n$}           (a0)
            (a0) edge              node [below]            {$c_{\mathsf{acc}}$} (a3)
            (a3) edge              node [below]            {$\mathsf{init}$}  (a4)
            (a4) edge [bend right] node [below left=-2mm]  {$c_{0}$}        (a5)
                 edge [bend right] node [right=0mm]        {$c_{1}$}        (a6)
                 edge              node [below]            {$c_{\triangleright}$} (a7)
            (a5) edge [bend right] node [left]        {$0$}           (a4)
            (a6) edge [bend right] node [below right=-1.5mm] {$1$}           (a4)
            (a7) edge              node [below]            {$\#$}             (a8);

  \path[loosely dotted,very thick, shorten >=3mm, shorten <=3mm,-,draw] (a1) -- (a2);
  \path[<-,draw] (a0) -- ++(-.8,0);
  \path[->,draw] (a8) -- ++(.8,0);

  \node[state]           (b0)  [right=5.5cm of a4]    {};
  \node                  (bx1) [above left of=b0]   {};
  \node                  (bx2) [above right of=b0]  {};
  \node                  (bx3) [above=6.5mm of bx1] {};
  \node                  (bx4) [above=6.5mm of bx2] {};
  \node[state]           (b1)  [above left of=bx1]  {$w$};
  \node[state]           (b2)  [above right of=bx2] {$v$};
  \node                  (bx5) [above of=b0]        {};
  \node[state,accepting] (b3)  [above of=bx5]       {};
  \node[state]           (b4)  [right of=b2]        {}; 

  \path[->, shorten >=2mm] (b0) edge node (e1) {} node [left]  {$a_1$} (bx1) 
                                edge node (e2) {} node [below right] {$a_n$} (bx2);
  \path[->, shorten <=2mm] (bx1) edge node (e3) {} node [left]  {$a_1$} (b1) 
                                 edge node (e4) {} node [right] {$a_n$} (bx3.south east)
                           (bx2) edge node (e5) {} node [left]  {$a_1$} (bx4.south west) 
                                 edge node (e6) {} node [above left=-1mm,near end] {$a_n$} (b2);
  \path[<-,draw] (b0) -- ++(-.8,0);

  \path[-,loosely dotted,very thick, shorten >=0mm, shorten <=0mm,-,draw] (e1) edge (e2) (e3) edge (e4) (e5) edge (e6);
  \path[loosely dotted,very thick, shorten >=-1mm, shorten <=-1mm,-,draw] (bx1.south east) -- (bx1.north west);
  \path[loosely dotted,very thick, shorten >=-1mm, shorten <=-1mm,-,draw] (bx2.south west) -- (bx2.north east);

  \path[->, draw, shorten >=-2mm] (b1) .. controls ++(-.4,-.7) .. node [below,sloped] {$c_{\hat{\delta}(w)}$} (bx5.south west);

  \path[->] (b1) edge              node [near start,above,sloped] {$\overline{c_{\hat{\delta}(w)}}$} (b3)
            (b2) edge              node [near start,above,sloped] {$\overline{c_{\hat{\delta}(v)}}$} (b3)
            (b3) edge [loop below] node [left] {$\Sigma_1$} node [right] {$\Sigma_2$}           ()
            (b2) edge              node [above]        {$c_{\mathsf{acc}}$}              (b4);
  \path[->] (b4) edge ++(.8,0);

  \node () [above right of=a1] {$\mathcal{A}_{\mathsf{chk}}$};
  \node () [below right of=b2] {$\mathcal{B}_{\mathsf{chk}}$};
\end{tikzpicture}
\end{center}
\caption{Two automata used to simulate computations of the Turing machine $\mathcal{M}$.}
\label{fig:simulation}
\end{figure}

\begin{proof}
  The automaton $\mathcal{A}_{\mathsf{chk}}$ is shown in
  Fig.~\ref{fig:simulation} on the left, assuming that
  $\Sigma_1 \setminus \{\mathsf{acc},\triangleleft,\triangleright,\#\}
  = \{a_1,\ldots,a_{n'}\}$, $n = n'-3$, $a_{n-2} = \triangleleft$,
  $a_{n-1} = \triangleright$, $a_n = \#$. Its structure forces
  \spoiler to do the following: if he wants to put a symbol from
  $\Sigma_1$ into the $\omega$-buffer then he first has to announce
  this by playing the corresponding $\Sigma_2$-copy. This tells
  \duplicator immediately what the next symbol in the $\omega$-buffer
  will be because of this synchronisation via the
  $0$-buffer. Moreover, in order for \spoiler to form a run that
  passed through this automaton infinitely often and for this run to
  be accepting, \spoiler has to eventually announce via
  $c_{\mathsf{acc}}$ that he would put $\mathsf{acc}$ into the
  $\omega$-buffer. However, at this point he actually puts
  $\mathsf{init}$ in instead and afterwards can only put in letters
  from $\Gamma$. This way he immediately sets up the buffer for the
  next simulation; remember that $\mathcal{M}$ is assumed to halt in
  $\triangleright\, w\, \mathsf{acc}\, \triangleleft$ when started in
  $\triangleright\, w\, \#\, v\, \triangleleft$. The same trick of
  announcing a symbol from a final configuration but putting a
  different one into the $\omega$-buffer is used to turn the
  end-marker $\triangleright$ into $\#$ in order to set up the buffer
  for the start of the next simulation.

  $\mathcal{B}_{\mathsf{chk}}$ is more difficult to depict. It is
  sketched on the right of Fig.~\ref{fig:simulation}. Its initial
  state is followed by a tree of depth $4$ that is used to read words
  of the form $a_1a_2a_3a_4$ from the $\omega$-buffer. This is used by
  \duplicator to remember the first $4$ symbols from the
  $\omega$-buffer which is supposed to be the beginning of a
  configuration of $\mathcal{M}$. Now she starts to accept
  synchronisation letters given by \spoiler who has begun to construct
  the next configuration. Remember that he can only play a
  synchronisation action $c_a$ if he puts $a$ into the $\omega$-buffer
  right away. This way \duplicator can control that he does indeed
  construct the valid successor configuration.
 
  Each state at depth $4$ of this tree that can be reached by reading
  the word $w = a_1a_2a_3a_4 \in (\Sigma_1)^4$ has two successors:
  with $c_{\hat{\delta}(w)}$ it can reach the state corresponding to
  reading $a_2a_3a_4$ one level below in the tree. This is used when
  \spoiler correctly chooses the next symbol of the unique successor
  configuration and transmits this through $c_{\hat{\delta}(w)}$. This
  is shown on the leftmost state of the tree structure. The rightmost
  state corresponding to, say $v$, shows an exception: if
  $\hat{\delta}(v) = \mathsf{acc}$ then \spoiler is about to produce
  the last configuration of $\mathcal{M}$'s computation, then
  \duplicator can move to the right and finish the simulation. The
  $\omega$-buffer is then set up for the next simulation already since
  \spoiler puts $\mathsf{init}$ instead of $\mathsf{acc}$ into the
  $\omega$-buffer.

  The other successor of the states corresponding to reading
  $w \in (\Sigma_1)^4$ from the $\omega$-buffer is reached with
  $\Sigma_2 \setminus \{ c_{\hat{\delta}(w)} \}$, here abbreviated as
  $\overline{c_{\hat{\delta}(w)}}$. This corresponds to \spoiler
  producing a symbol that is not the next one in the unique successor
  configuration, and this takes \duplicator to a state that makes
  \spoiler lose. Formally, $\mathcal{B}_{\mathsf{chk}}$ has states
  $(\Sigma_1)^{\le 4} \cup \{\mathsf{sink},\mathsf{done}\}$ and the following
  transitions  with initial state $\epsilon$. \vspace{-0.3cm}%
  \begin{align*}
    \Transition{w}{a}{wa} &, \text{ for } w \in (\Sigma_1)^{\le 3}, a \in \Sigma_1  &
    \Transition{w}{c_{\mathsf{acc}}}{\mathsf{done}} &, \text{ for } w \in (\Sigma_1)^4 
    \hspace{1cm} 
    \Transition{aw}{c_{\hat{\delta}(aw)}}{w} &, \text{ for } w \in (\Sigma_1)^3 \\
    \Transition{w}{c}{\mathsf{sink}} &, \text{ for } w \in (\Sigma_1)^4, c \not\in \{ c_{\hat{\delta}(w)}, c_{\mathsf{acc}} \}  & 
    \Transition{\mathsf{sink}}{x}{\mathsf{sink}} &, \text{ for every } x \in \Sigma_1 \cup \Sigma_2 & & \qedhere
  \end{align*}  %
\end{proof}

\begin{theorem}
  \label{thm:RBG2simu}
  Given a RBG $G$, one can construct two NBA $\mathcal{A}$ and
  $\mathcal{B}$ such that
  $\mathcal{A} \DKsimu{(\omega,0)}{(\Sigma_1,\Sigma_2)} \mathcal{B}$
  iff player $0$ has a winning strategy for $G$.
\end{theorem}

\begin{proof}
  We construct $\mathcal{A}$ and $\mathcal{B}$ by plugging (variants
  of) $\mathcal{A}_{\mathsf{chs}}$, $\mathcal{A}_{\mathsf{frc}}$,
  $\mathcal{A}_{\mathsf{chk}}$, $\mathcal{B}_{\mathsf{frc}}$ and
  $\mathcal{B}_{\mathsf{chk}}$ together as follows. First of all, we
  modify $\mathcal{A}_{\mathsf{chk}}$ and $\mathcal{B}_{\mathsf{chk}}$
  such that they remember if the buffer content that is produced in
  the end according to Lemma~\ref{lem:simucheck} is
  $\#\,1w\,\mathsf{init}\,\triangleleft$ or
  $\#\,0w\,\mathsf{init}\,\triangleleft$ for some $w \in
  \{0,1\}^*$. Thus, they have two different states each that would be
  marked with outgoing edges. Intuitively, they distinguish the case
  in which player $0$, resp.\ player $1$ is the owner of the current
  RBG node and therefore needs to perform a choice next.
  \vspace*{-3mm}

  Moreover, we use an automaton $\mathcal{A}'_{\mathsf{chk}}$ that is
  obtained from $\mathcal{A}_{\mathsf{chk}}$ by appending
  \begin{tikzpicture}[every state/.style={minimum size=3mm,inner sep=0pt}, node distance=13mm,baseline=-1mm]
    \node                  (a0) {};
    \node[state,accepting] (a1) [right of=a0] {};
    \path[->] (a0) edge node [above] {$c_{\mathsf{rej}}$} (a1)
                    (a1) edge [loop right] node [right] {$\mathsf{rej}$} ();
  \end{tikzpicture}
  to its initial state. Remember that this automaton is used by
  \spoiler to produce a computation of the Turing machine
  $\mathcal{M}$ which is checked for being valid and accepting by
  \duplicator with $\mathcal{B}_{\mathsf{chk}}$. With
  $\mathcal{A}_{\mathsf{chk}}$, \spoiler has a strategy not to lose if
  in the game with initial buffer content
  $\triangleright\, w\, \# \, v\, \mathsf{init}\, \triangleleft$ if
  $w$ is a valid successor of node $v$. Thus, it can be used to verify
  \spoiler's choices of such a successor node that were made with
  $\mathcal{A}_{\mathsf{chs}}$. However, if \duplicator proposed $w$
  instead, and it is not a valid successor of $v$ in $G$, then
  \duplicator should lose. This happens using
  $\mathcal{A}'_{\mathsf{chk}}$: the simulation of $\mathcal{M}$'s
  computation will ultimately reach a rejecting configuration which
  allows \spoiler to move to the accepting sink which \duplicator
  cannot match.

  The following picture shows how $\mathcal{A}$ (left) and
  $\mathcal{B}$ (right) are obtained. A dashed line is drawn in order
  to indicate that outgoing states (on the right in the automata) are
  merged with incoming states (on the left). For those automata that
  have two outgoing states we use the convention that the upper one is
  used when the next chosen node belongs to player $1$ and the lower
  one otherwise.
  
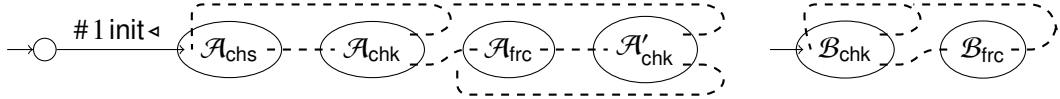
\begin{center}
    \begin{tikzpicture}[every state/.style={minimum size=3mm,inner sep=0pt}, node distance=13mm]
      \tikzstyle{automaton}=[shape=ellipse,inner sep=3pt,draw]

      \node[automaton] (achs)                       {$\mathcal{A}_{\mathsf{chs}}$};
      \node[automaton] (achk1) [right=5mm of achs]  {$\mathcal{A}_{\mathsf{chk}}$};
      \node[automaton] (afrc)  [right=5mm of achk1] {$\mathcal{A}_{\mathsf{frc}}$};
      \node[automaton] (achk2) [right=5mm of afrc]  {$\mathcal{A}'_{\mathsf{chk}}$};
      \node[state]     (a0)    [left=16mm of achs]  {};

      \path[thick,dashed,-,draw, shorten >=-2mm, shorten <=-2mm] (achs.east) edge (achk1.west) (afrc.east) edge (achk2.west);
      \path[thick,dashed,-,draw, shorten <=-2mm, rounded corners] ($(achk1.east)+(0,-.2)$) -- ++(.2,0) -- ++(.2,.2) -- ($(afrc.west)+(.2,0)$);
      \path[thick,dashed,-,draw, shorten <=-2mm, rounded corners] ($(achk1.east)+(0,.2)$) -- ++(.2,0) -- ++(.2,.2) -- ++(-.2,.2) -- ++(-3.3,0) -- ++(0,-.3)
      -- ($(achs.west)+(.2,0)$);
      \path[thick,dashed,-,draw, shorten <=-2mm, rounded corners] ($(achk2.east)+(0,.2)$) -- ++(.2,0) -- ++(.2,.2) -- ++(-.2,.2) -- ++(-3.5,0);
      \path[thick,dashed,-,draw, shorten <=-2mm, rounded corners] ($(achk2.east)+(0,-.2)$) -- ++(.2,0) -- ++(.2,-.2) -- ++(-.2,-.2) -- ++(-3.4,0) -- ++(0,.4);

      \path[->, shorten >=-1mm] (a0) edge node [above] {$\#\,1\,\mathsf{init}\,\triangleleft$} (achs.west);
      \path[<-,draw] (a0) -- ++(-.5,0);

      \node[automaton] (bchk) [right=13mm of achk2] {$\mathcal{B}_{\mathsf{chk}}$};
      \node[automaton] (bfrc) [right=6mm of bchk] {$\mathcal{B}_{\mathsf{frc}}$};

      \path[<-,draw,shorten <=-1mm] (bchk) -- ++(-1,0);
      \path[thick,dashed,-,draw, shorten <=-2mm, rounded corners] ($(bchk.east)+(0,-.2)$) -- ++(.2,0) -- ++(.2,.2) -- ($(bfrc.west)+(.2,0)$);
      \path[thick,dashed,-,draw, shorten <=-2mm, rounded corners] ($(bchk.east)+(0,.2)$) -- ++(.2,0) -- ++(.2,.2) -- ++(-.2,.2) -- ++(-1.4,0) -- ++(0,-.3)
      -- ($(bchk.west)+(.2,0)$);
      \path[thick,dashed,-,draw, shorten <=-2mm, rounded corners] (bfrc.east) -- ++(.2,0) -- ++(.2,.2) -- ++(0,.2) -- ++(-.2,.2) -- ++(-1.7,0);

    \end{tikzpicture}
  \end{center}
  The part at the beginning in $\mathcal{A}$ ensures that the
  $\omega$-buffer is filled with $\#\,1\,\triangleleft$, i.e.\ the
  initial game node which is owned by \spoiler. He then uses
  $\mathcal{A}_{\mathsf{chs}}$ to choose a successor, leading to a
  buffer content of the form
  $\triangleright\, w\, \#\, 1\, \mathsf{init}\, \triangleleft$. The
  play then proceeds in $\mathcal{A}_{\mathsf{chk}}$ which requires
  \duplicator to make moves in $\mathcal{B}_{\mathsf{chk}}$, finally
  leading to the buffer content $\#\, w\,\mathsf{init}\,\triangleleft$
  according to Lemma~\ref{lem:simucheck}. Depending on whether $w$
  starts with $1$ or $0$, the corresponding player makes choices to
  fill the buffer to a content of
  $\triangleright\, v\, \#\, w\, \mathsf{init}\, \triangleleft$,
  either \spoiler using $\mathcal{A}_{\mathsf{chs}}$ or \duplicator
  using $\mathcal{B}_{\mathsf{frc}}$ with \spoiler executing her
  wishes in $\mathcal{A}_{\mathsf{frc}}$, according to
  Lemma~\ref{lem:Dupforce}.

  Finally, we need to make sure that \duplicator wins iff the
  underlying play in $G$ visits states of the form $w1$ infinitely
  often.  We make the last states of $\mathcal{A}_{\mathsf{chs}}$ and
  $\mathcal{A}_{\mathsf{frc}}$ final. Then any run of \spoiler in
  which he produces finite simulations of the Turing machine and
  performs choices of finite nodes in the RBG, is accepting. Hence, we
  need to give \duplicator the ability to answer with an accepting run
  if it corresponds to a play in the RBG that was winning for her
  because it visits infinitely many states of the form
  $\{0,1\}^*1$. This can easily be done by letting her go through an
  accepting state in $\mathcal{B}_{\mathsf{chk}}$ only if \spoiler has
  signalled to her that the last configuration of the simulated
  computation is of the form
  $\triangleright\, w\, 1\, \mathsf{acc}\, \triangleleft$ which can
  easily be checked by adding a few more states to this automaton. At
  last, we also need to give \spoiler the ability to win when
  \duplicator does not do her part in the simulation process
  properly. This can occur in $\mathcal{B}_{\mathsf{frc}}$ which
  \duplicator should use to force \spoiler to put a finite
  $\triangleright\, w$ for $w \in \{0,1\}^+$ into the buffer. So we
  need to make sure that she eventually terminates this forcing
  process. This is easily done by making the first state in
  $\mathcal{A}_{\mathsf{frc}}$ accepting for \spoiler.
\end{proof}

\begin{remark}
  Given a RBG $G$, one can construct two NBA $\mathcal{A}$ and
  $\mathcal{B}$ over the alphabet $\{a,b,c\}$ such that
  $\mathcal{A} \DKsimu{(\omega,0)}{(\{a,b\},\{c\})} \mathcal{B}$ iff
  player $0$ has a winning strategy for $G$.
\end{remark}

\begin{proof}
  Consider the automaton $\mathcal{A}$ constructed in the proof
  above. The crucial property is that this automaton can at most
  perform two \emph{consecutive} transitions labeled in
  $\Sigma_2$. Now enumerate all nonempty words over $\Sigma_2$ of
  length at most $2$ as $w_1,w_2,\dots,w_n$. The automaton
  $\mathcal{A'}$ is obtained from $\mathcal{A}$ by deleting all
  transitions labeled in $\Sigma_2$ and replacing any path of labeled
  by a word $w_i$ by a path labeled by $c^i$. Doing the analogous
  changes in $\mathcal{B}$, we obtain $\mathcal{B'}$. Now it should be
  clear that
  $\mathcal{A'}\DKsimu{(\omega,0)}{(\Sigma_1,\{c\})}\mathcal{B'}$
  holds if and only if
  $\mathcal{A}\DKsimu{(\omega,0)}{(\Sigma_1,\Sigma_2)}\mathcal{B}$,
  i.e., if and only if player~0 wins $G$.
\end{proof}

Putting this together with the lower bound for solving recursive
B\"uchi games established above, we obtain the following result that
is in stark contrast to the EXPTIME-decidability of $\DKsimu{\omega}{\Sigma}$.

\begin{corollary}
  The relation $\DKsimu{(\omega,0)}{(\{a,b\},\{c\})}$ is
  $B\Sigma^1_1$-hard.
\end{corollary}

\begin{proof}
  All that is left is a coding of the alphabet $\Sigma_1$ by words
  over the binary alphabet $\{a,b\}$.
\end{proof}

Clearly, the same lower bound holds for any multi-buffer simulation
involving at least two buffers of which at least one is unbounded such
that $|\Sigma_1\setminus\Sigma_2|\ge2$ and
$\Sigma_2\setminus\Sigma_1\neq\emptyset$.

However, if $|\Sigma_1\setminus\Sigma_2| = 1$, then the multi-buffer
game is decidable. 

\begin{theorem}
\label{thm:omega_zero_decidable}
  The relation $\DKsimu{(\omega,0)}{(\{(a)\},\{c,d\})}$ is decidable.
\end{theorem}
\begin{proof}
  Let
  $\mathcal{A} =
  (Q^{\mathcal{A}},\Sigma,p_I,\delta^{\mathcal{A}},F^{\mathcal{B}})$
  and
  $\mathcal{B} =
  (Q^{\mathcal{B}},\Sigma,q_I,\delta^{\mathcal{B}},F^{\mathcal{B}})$
  be two NBA over the alphabet $\{a,c,d\}$. 
  Consider the following 1-counter automaton: The set of states equals
    $Q^{\mathcal{A}}\times
    Q^{\mathcal{B}}\times\{c,d,\varepsilon\}\times\{0,1\}$, and
 all transitions are $\epsilon$-transitions.
  For all transitions $\Transition{p}{a}{p'}$ in $\mathcal{A}$,
    the 1-counter automaton has an $\epsilon$-transition from
    $(p,q,\epsilon,0)$ to $(p',q,\epsilon,1)$ that increments the
    counter.
  For all transitions $\Transition{p}{c}{p'}$ in $\mathcal{A}$,
    the 1-counter automaton has an $\epsilon$-transition from
    $(p,q,\epsilon,0)$ to $(p',q,c,1)$ that leaves the counter
    unchanged (and similarly for transitions $\Transition{p}{d}{p'}$).
   From any state $(p,q,\epsilon,1)$, there is an
    $\epsilon$-transition to $(p,q,\epsilon,0)$ that does not change
    the content of the counter.
   For all transitions $\Transition{q}{a}{q'}$ in $\mathcal{B}$,
    the 1-counter automaton has $\epsilon$-transitions from
    $(p,q,\epsilon,1)$ to $(p,q',\varepsilon,0)$ and to
    $(p,q',\varepsilon,1)$ that decrement the counter.
   For all transitions $\Transition{q}{c}{q'}$ in $\mathcal{B}$,
    the 1-counter automaton has $\epsilon$-transitions from
    $(p,q,c,1)$ to $(p,q',\varepsilon,0)$ and to
    $(p,q',\varepsilon,1)$ that leave the counter unchanged (and
    similarly for transitions $\Transition{q}{d}{q'}$).
 
  The set of configurations of this 1-counter automaton equals
  $
    Q^{\mathcal{A}}\times
    Q^{\mathcal{B}}\times\{c,d,\varepsilon\}\times\{0,1\}     
    \times\Nat\,.
  $
  We now define a game whose nodes are the configurations:
  Configurations of the form $(p,q,x,0,n)$ belong to player~0,
    the other configurations belong to player~1.
  The moves are given by the transitions of the 1-counter
    automaton.
  A finite play is lost by the player that got stuck. An
    infinite play $(p_i,q_i,x_i,n_i)_{i\ge0}$ is won by player~1 if
  $p_i\in F^{\mathcal{A}}$ for only finitely many $i\in\Nat$ or
  $q_i\in F^{\mathcal{B}}$ for infinitely many $i\in\Nat$ and
      the counter increases only finitely often or decreases
      infinitely often.

  Then player~1 wins this game $G$ iff \duplicator wins the game
  $\DKsimugame{(\omega,0)}{(\{a\},\{c,d\})}{\mathcal{A}}{\mathcal{B}}$.
  The existence of a winning strategy of player~1 can be expressed as
  a MSO-formula talking about the configuration graph of the 1-counter
  automaton. Since the MSO-theories of 1-counter automata are
  uniformly decidable \cite{MulS85}, the existence of a winning strategy for
  player~II in the game $G$ can be decided.
\end{proof}

\newcommand{\cA}{{\mathcal A}}
\newcommand{\cB}{{\mathcal B}}
\newcommand{\Run}{{\mathit{Run}}}
\newcommand{\ARun}{{\mathit{ARun}}}

\section{Completeness}
\label{sec:sufficiency}
In Section~\ref{sec:multibuf}, we have shown that if \duplicator wins
a multi-buffer game between $\mathcal{A},\mathcal{B}$ over the trace
alphabet $\sigma=(\Sigma_i)_{i\in[k]}$, then
$[L(\mathcal{A})]_\sigma \subseteq [L(\mathcal{B})]_\sigma$.  In this
section, we characterize the relation $\DKsimu{\kappa}{\sigma}$ for
unbounded buffers analogously to the one-buffer case
\cite{DBLP:journals/corr/HutagalungLL14}. More precisely, we will show
that unbounded multi-buffer simulation is equivalent to the existence
of a continuous function that maps accepting runs of $\mathcal A$ to
accepting runs of $\mathcal B$ with trace-equivalent words.\medskip

\noindent\textbf{Notation.} Throughout this section, let $k\in\Nat$ be fixed,
let $\kappa_\omega$ be the capacity function for $k$ buffers with
$\kappa_\omega(i)=\omega$ for all $i\in[k]$, and let
$\sigma=(\Sigma_i)_{i\in[k]}$ be an arbitrary trace
alphabet. Furthermore, we fix two automata
$\cA=(Q^\cA,\Sigma,p_{\mathsf{I}},\delta^\cA,F^\cA)$ and
$\cB = (Q^\cB,\Sigma,q_{\mathsf{I}},\delta^\cB,F^\cB)$. \medskip

First recall that an infinite run of some NBA $\mathcal{A}$ is an
infinite word over $Q^\mathcal{A}\cup \Sigma$.  We denote the set of
runs of $\mathcal A$ by $\mathit{Run}(\mathcal{A})$ and the set of
accepting runs by $\mathit{ARun}(\mathcal{A})$.

Given a set $\Delta$, the set $\Delta^{\omega}$ of infinite words over
$\Delta$ is equipped with the standard structure of a metric
space. The distance $d(x,y)$ between two distinct infinite sequences
$x_0x_1x_2\dots$ and $y_0y_1y_2\dots$ is
$\inf \{ 2^{-i} \mid x_j = y_j$ for all $j < i \}$.
Intuitively, two words are ``significantly close'' if they share a
``significantly long'' prefix.

We call a function
$f\colon\mathit{ARun}(\mathcal{A})\to \mathit{ARun}(\mathcal{B})$
\emph{trace-preserving} if for all accepting runs
$\rho \in \mathit{ARun}(\mathcal{A})$, the run $f(\rho)$ is accepting
and the words of these two runs are trace equivalent. Consequently,
$[L(\mathcal{A})]_\sigma \subseteq [L(\mathcal{B})]_\sigma$ iff there
exists a trace-preserving function
$f\colon \mathit{ARun}(\mathcal{A}) \to
\mathit{ARun}(\mathcal{B})$. Throughout this section, we will show
that the existence of a continuous and trace-preserving function
characterizes
$\mathcal{A} \DKsimu{\kappa_\omega}{\sigma} \mathcal{B}$.

\begin{lemma}
  \label{lem:wintocontinuous}
  If $\mathcal{A} \DKsimu{\kappa_\omega}{\sigma} \mathcal{B}$, then
  there exists a continuous trace-preserving function
  $f\colon \mathit{ARun}(\mathcal{A}) \to \mathit{ARun}(\mathcal{B})$.
\end{lemma}

\begin{proof}
  Suppose \duplicator wins
  $\DKsimugame{\kappa_\omega}{\sigma}{\mathcal{A}}{\mathcal{B}}$ with
  some winning strategy $\theta$. We define $f(\rho) = \rho'$ such
  that $\rho,\rho'$ are the output of
  $\DKsimugame{\kappa_\omega}{\sigma}{\mathcal{A}}{\mathcal{B}}$, in
  which \spoiler plays $\rho$ and \duplicator plays according
  to~$\theta$. The function $f$ is trace-preserving since $\theta$ is
  a winning strategy for \duplicator.

  Let $\rho_1\in\ARun(\cA)$ be an accepting run of $\cA$ and let
  $n\in\Nat$. Since $f(\rho_1)\in\ARun(\cB)$ is the output of
  \duplicator's moves according to the winning strategy $\theta$,
  there is some round $m$ such that \duplicator's output after that
  round has length at least $n$. Now let $\rho_2$ be another run of
  $\cA$ that agrees with $\rho_1$ in the first $m$ transitions, i.e.,
  with $d(\rho_1,\rho_2)\le 2^{-(2m+3)}$. Then, in the first $m$
  rounds, \duplicator does not see any difference between \spoiler's
  runs $\rho_1$ and $\rho_2$. Hence \duplicator (playing according to
  the strategy $\theta$) makes the same moves. This implies that
  $f(\rho_1)$ and $f(\rho_2)$ agree in the first $n$ positions, i.e.,
  $d(f(\rho_1),f(\rho_2))< 2^{-n}$. Thus, we showed that $f$ is
  continuous.
\end{proof}

\begin{example}
  Under the assumption
  $\mathcal{A} \DKsimu{\kappa_\omega}{\sigma} \mathcal{B}$, of the
  lemma, there need not be a continuous function
  $f\colon\Run(\cA)\to\Run(\cB)$ (let alone a continuous function from
  $(Q^\cA\cup\Sigma)^\omega$ to $(Q^\cB\cup\Sigma)^\omega$) that is
  trace-preserving and maps accepting runs to accepting runs.
  \vspace*{-3.4mm}
  
  \noindent
  \parbox{0.69\textwidth}{ \quad\enspace 
  Consider these two NBA $\mathcal{A}$ (above) and $\mathcal{B}$
  (below), and the trace alphabet $\sigma= (\{a,b\})$. We have
  $\mathcal{A} \DKsimu{\kappa_\omega}{\sigma} \mathcal{B}$, since
  \duplicator has the following winning strategy: if, in state $q_0$
  or $q_1$, \duplicator sees $a$ in the buffer, then she moves to
  $q_1$. If, in state $q_0$ or $q_2$, \duplicator sees $b$ in the
  buffer, then she moves to $q_2$. In all other cases, \duplicator
  skips her move. This is a winning strategy since the only accepting
  run $(p_0a)^\omega$ of $\mathcal {A}$ is answered by the accepting
  run $q_0a\,(q_1a)^\omega$ of $\mathcal{B}$.}
   \raisebox{3mm}{\begin{minipage}{0.29\textwidth}
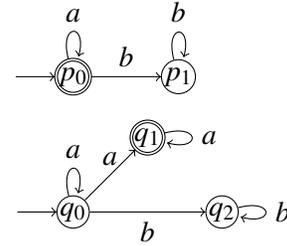
\begin{tikzpicture}[initial text={}, node distance=14mm, every state/.style={minimum size=3mm, inner sep=0pt}]
  \begin{scope}[xshift=-3cm]
  \node[state,initial,accepting]   (0)   {$p_0$};
  \node[state]                     (1) [right of=0] {$p_1$};
  \end{scope}
  \begin{scope}[xshift=3cm]
  \node[state,initial]   (2) [below of=0, node distance=1.8cm] {$q_0$};
  \node[state,accepting] (3) [above right of=2] {$q_1$};
  \node[state]           (4) [below right of=3] {$q_2$};
  \end{scope}

  \path[->] (0) edge[loop above] node [above] {$a$} (0)
                edge node [above] {$b$}              (1)

            (1) edge[loop above] node  {$b$}        (1)

            (2) edge node [above] {$a$}              (3)
                edge node [below] {$b$}              (4)
                edge[loop above] node [above] {$a$}  (2)

            (3) edge[loop right] node {$a$} (3)
            (4) edge[loop right] node {$b$} (4);
\end{tikzpicture} \end{minipage}} \\

  Now let $f\colon\Run(\cA)\to\Run(\cB)$ be a continuous and
  trace-preserving function that maps accepting runs of $\cA$ to
  accepting runs of $\cB$. Since $f$ is trace-preserving, we get
  $f((p_0a)^m p_0b (p_1b)^\omega)=(q_0a)^m q_0b (q_2b)^\omega$ for all
  $m\in\Nat$. Note that, when $m$ grows, the runs
  $(p_0a)^m p_0b (p_1b)^\omega$ converge to the run $(p_0a)^\omega$
  and their $f$-images $(q_0a)^m q_0b (q_2b)^\omega$ converge to
  $(q_0a)^\omega$. Since $f$ is continuous, this implies
  $f((p_0a)^\omega)=(q_0a)^\omega$, i.e., the accepting run
  $(p_0a)^\omega$ of $\cA$ is mapped to the non-accepting run
  $(q_0a)^\omega$ of $\cB$. Hence, indeed, there is no continuous and
  trace preserving function $f\colon\Run(\cA)\to\Run(\cB)$ that maps
  accepting runs to accepting runs.
\end{example}

We are interested in the reverse direction of
Lemma~\ref{lem:wintocontinuous} because it shows that continuity is
the weakest condition that implies multi-buffer simulation on
unbounded buffers.  We will use a \emph{delay game} as in
\cite{DBLP:journals/corr/abs-1209-0800}.  In this game, the winning
condition is given by some function
$f\colon \mathit{ARun}(\mathcal{A}) \to \mathit{ARun}(\mathcal{B})$, and
\duplicator is allowed to form a run without considering any buffers.

Let $\mathcal{A}, \mathcal{B}$ be NBAs and
$f\colon \mathit{ARun}(\mathcal{A}) \to \mathit{ARun}(\mathcal{B})$ be
a function that maps accepting runs of $\mathcal{A}$ to accepting runs
of $\mathcal{B}$.  The \emph{delay game}
$\delaygame{f}{\mathcal{A}}{\mathcal{B}}$ is played between players
\spoiler and \duplicator on $\mathcal{A}$ and $\mathcal{B}$, where in
each round each player tries to extend a run.  A configuration is a
pair of finite runs $(r_{\mathcal{A}},r_{\mathcal{B}})$ on
$\mathcal{A}$ and $\mathcal{B}$, respectively (the pair
$(p_{\mathsf{I}},q_{\mathsf{I}})$ is the initial configuration).   For every
round $i >0 $, with configuration $(r_\mathcal{A},r_\mathcal{B})$:
\spoiler tries to extend $r_{\mathcal{A}}$ with one step:
  $r'_\mathcal{A} := r_\mathcal{A}ap$, and
\duplicator tries to extend $r_\mathcal{B}$ with $n \geq 0$
  steps: $r'_{\mathcal{B}} := r_{\mathcal{B}}b_1q_1 \ldots b_np_n$.
They continue to the next round with the configuration
$(r'_\mathcal{A},r'_\mathcal{B})$.

A play builds two runs: an infinite run $\rho$ of $\cA$ (chosen by
\spoiler) and a finite or infinite run $\rho'$ of $\cB$ (chosen by
\duplicator). We say that \duplicator wins the play iff $\rho$ is not
accepting or $f(\rho)=\rho'$. We write
$\mathcal{A} \sqsubseteq^f_{\mathsf{del}} \mathcal{B}$ as shorthand
for ``\duplicator has a winning strategy in the delay game
$\delaygame{f}{\mathcal{A}}{\mathcal{B}}$''.

Note that if $f\colon \mathit{ARun}(\mathcal{A}) \to \mathit{ARun}(\mathcal{B})$ is
continuous, then $\mathcal{A} \sqsubseteq^f_{\mathsf{del}} \mathcal{B}$.
The winning strategy for \duplicator is to move properly according to $f$:
on configuration $(r_\mathcal{A},r_\mathcal{B})$,
if there
  exists $(b,q) \in \Sigma\times Q^\cB$ such that for any
  $\rho \in \ARun(\mathcal{A})$ with prefix $r_{\mathcal{A}}$,
  $r_{\mathcal{B}}bq$ is a prefix of $f(\rho)$, then \duplicator
  extends $r_{\mathcal{B}}$ to $r'_{\mathcal{B}} := r_{\mathcal{B}}bq$.  
  Otherwise, \duplicator   skips her move.
In this way, \duplicator always forms a run that is the image of $f$.
  
\begin{lemma}
  \label{lem:continuoustowin}
  Let
  $f\colon \mathit{ARun}(\mathcal{A}) \to \mathit{ARun}(\mathcal{B})$
  be a continuous function. Then
  $\mathcal{A} \sqsubseteq^f_{\mathsf{del}} \mathcal{B}$.
\end{lemma}

Suppose $k=1$, i.e., we are in the one-buffer case and let
$f\colon\ARun(\cA)\to\ARun(\cB)$ be continuous and
trace-preserving. Then, as we saw above, \duplicator has a winning
strategy~$\theta$ in the delay game
$\delaygame{f}{\mathcal{A}}{\mathcal{B}}$. In
\cite{DBLP:journals/corr/HutagalungLL14}, it was shown that this strategy is also
a winning strategy in the buffer game
$\DKsimugame{\kappa_\omega}{(\Sigma)}{\cA}{\cB}$. The following
example shows that this is not the case in the multi-buffer game:
$\theta$ may tell \duplicator to output a letter that has not already
been played by \spoiler.

\begin{example}
  \label{ex:notbuffered}
  Consider the following two NBA $\mathcal{A}$ (left) and
  $\mathcal{B}$ (right) and the trace alphabet
  \vspace*{.5mm}

\noindent
\begin{minipage}{0.68\textwidth} \setlength{\parfillskip}{0pt}
  $\sigma=(\{a\},\{b\})$.  Take the continuous function
  $f\colon \mathit{ARun}(\mathcal{A}) \to
  \mathit{ARun}(\mathcal{B})$ that maps all infinite runs of
  $\mathcal{A}$ to the unique infinite run of
  $\mathcal{B}$. \duplicator wins
  $\delaygame{f}{\mathcal{A}}{\mathcal{B}}$ with strategy $\theta$:
\end{minipage}\enspace
\raisebox{2mm}{\begin{minipage}{0.3\textwidth}
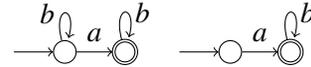
\begin{tikzpicture}[initial text={}, node distance=8mm, every state/.style={minimum size=3mm, inner sep=0pt}]
  \node[state,initial]   (qa0)                      {};
  \node[state,accepting] (qa1) [right of=qa0]       {};
  
  \path[->] (qa0) edge [loop above] node [left]      {$b$}        (qa0)
            (qa0) edge              node [above]      {$a$}        (qa1)
            (qa1) edge [loop above] node [right]      {$b$}        (qa1);
            
  \node[state,initial]   (qb0) [right of=qa0, node distance = 2.2cm] {};
  \node[state,accepting] (qb1) [right of=qb0]                      {};

  \path[->] (qb0) edge              node [above]      {$a$}        (qb1)
            (qb1) edge [loop above] node [right]      {$b$}        (qb1);   

\end{tikzpicture}
\end{minipage}} \\[.5mm]
in the first round she forms $q_0aq_1$, in the second round she forms
$q_0aq_1bq_1$, and so on.  However, \duplicator cannot use $\theta$ to
win $\DKsimugame{\kappa_\omega}{\sigma}{\mathcal{A}}{\mathcal{B}}$,
since \spoiler may not output $a$ in the first round. Nevertheless,
\duplicator wins
$\DKsimugame{\kappa_\omega}{\sigma}{\mathcal{A}}{\mathcal{B}}$ by
simply waiting for the first $a$ in buffer 1 and then emptying both
buffers and, from then on, follows the strategy $\theta$.
\end{example}

More generally, we can derive a winning strategy for \duplicator on
$\DKsimugame{\kappa_\omega}{\sigma}{\mathcal{A}}{\mathcal{B}}$, from
some winning strategy $\theta$ on
$\delaygame{f}{\mathcal{A}}{\mathcal{B}}$, as stated in the following
lemma.

\begin{lemma}
  \label{lem:continuoustowin2}
  Let
  $f\colon \mathit{ARun}(\mathcal{A}) \to \mathit{ARun}(\mathcal{B})$
  be a continuous and trace-preserving function. Then
  $\cA\DKsimu{\kappa_\omega}{\sigma}\cB$.
\end{lemma}

\begin{proof}
  We naturally extend the projection functions
  $\pi_i\colon\Sigma^\infty\to{\Sigma_i}^\infty$ to
  $\pi_i\colon (Q\times\Sigma\times Q)^\infty\to{\Sigma_i}^\infty$ by
  setting $\pi_i(p,a,q)=\pi_i(a)$.

  By Lemma~\ref{lem:continuoustowin}, \duplicator has a winning
  strategy~$\theta$ in the delay game
  $\delaygame{f}{\mathcal{A}}{\mathcal{B}}$. To win the multi-buffer
  game $\DKsimugame{\kappa}{\sigma}{\mathcal{A}}{\mathcal{B}}$,
  \duplicator tries to mimic this strategy~$\theta$. Her problem is
  that in some situation, $\theta$ tells her to play some transition,
  but the letter of this transition is not available in the buffers.

  We next describe the modified strategy $\theta'$ of \duplicator:
  Suppose \spoiler has played the finite run $r_\cA$ in the delay game
  $\delaygame{f}{\mathcal{A}}{\mathcal{B}}$. Let $r$ be \duplicator's
  answer according to her strategy $\theta$. Let $r_\cB$ be the
  maximal prefix of $r$ such that, for all $i\in[k]$, the word
  $\pi_i(r_\cB)$ is a prefix of $\pi_i(r_\cA)$. Then \duplicator's
  strategy $\theta'$ shall ensure that she outputs this run $r_\cB$ in
  response to \spoiler playing $r_\cA$.

  We first verify that \duplicator can play according to this
  strategy: so suppose \spoiler extends his run $r_\cA$ to
  $r_\cA'=r_\cA a p$. Then \duplicator's answer $r'$ in the delay game
  extends the run $r$. The maximal prefix $r_\cB'$ of $r'$ such that
  $\pi_i(r_\cB')$ is a prefix of $\pi_i(r'_\cA)=\pi_i(r_\cA)\pi_i(a)$
  for $i\in[k]$ extends $r_\cB$ by some word
  $x\in(Q^\cB\cup\Sigma)^*$. Then, clearly \duplicator can play the
  difference $x$ between $r_\cB$ and $r_\cB'$.

  It remains to be shown that $\theta'$ is a winning strategy. To this
  aim, let $\rho_\cA$ be an accepting run of $\cA$. Since $\theta$ is
  winning in the delay game $\delaygame{f}{\cA}{\cB}$, the accepting
  run $\rho=f(\rho_\cA)$ is \duplicator's answer in this game to
  \spoiler playing $\rho_\cA$. Let $\rho_\cB$ be \duplicator's answer
  according to the strategy $\theta'$ in the simulation game
  $\DKsimugame{\kappa_\omega}{\sigma}{\cA}{\cB}$. We show
  $\rho=\rho_\cB$: Clearly, by the construction of $\theta'$, any
  finite prefix of $\rho_\cB$ is a prefix of $\rho$. Conversely, let
  $r_\cB$ be a finite prefix of $\rho_\cB$. There is a finite prefix
  $r_\cA$ of $\rho_\cA$ such that, once \spoiler has played $r_\cA$ in
  the simulation game, \duplicator's answer according to $\theta'$ is
  at least $r_\cB$. The rules of the simulation game imply
  $\pi_i(r_\cB)\le\pi_i(r_\cA)$ for all $i\in[k]$. 

  Since $f$ is trace-preserving, we get $\pi_i(\rho_\cA)=\pi_i(\rho)$
  for all $i\in[k]$. Hence we have, for all $i\in[k]$,
  $
     \pi_i(r_\cB) \le \pi_i(r_\cA) \le \pi_i(\rho_\cA) = \pi_i(\rho)\,.
  $
  It follows that there is a finite prefix $r$ or $\rho$ such that
  $
     \pi_i(r_\cB)\le\pi_i(r)
  $
  for all $i\in[k]$. Since both, $r_\cB$ and $r$ are prefixes of
  $\rho_\cB$, this implies that $r_\cB$ is a prefix of $r$ and
  therefore of $\rho$. Consequently,
  $\rho_\cB=\rho=f(\rho_\cA)$. Since $f$ maps the accepting run
  $\rho_\cA$ to an accepting run, $\rho_\cB$ is
  accepting. Furthermore, since $f$ is trace-preserving, the words of
  $\rho_\cA$ and $\rho_\cB$ are equivalent. Hence $\theta'$ is a
  winning strategy for \duplicator.
\end{proof}

Putting Lemmas~\ref{lem:wintocontinuous}, \ref{lem:continuoustowin}
and \ref{lem:continuoustowin2} together we obtain the following
characterisation of a case in which multi-buffer simulation is
complete for trace inclusion, namely that in which there is not only a
trace-preserving function between the runs but one that is
additionally continuous.

\begin{theorem}
  \label{thm:continuous}
  Let $\mathcal{A}$, $\mathcal{B}$ be two NBA over the trace alphabet
  $\sigma=(\Sigma_i)_{i\in[k]}$. We have
  $\mathcal{A} \DKsimu{\kappa_\omega}{\sigma} \mathcal{B}$ if and only
  if there exists a continuous trace-preserving function
  $f\colon \mathit{ARun}(\mathcal{A}) \to \mathit{ARun}(\mathcal{B})$.
\end{theorem}

\section{Conclusion and Further Work}
\label{sec:concl}

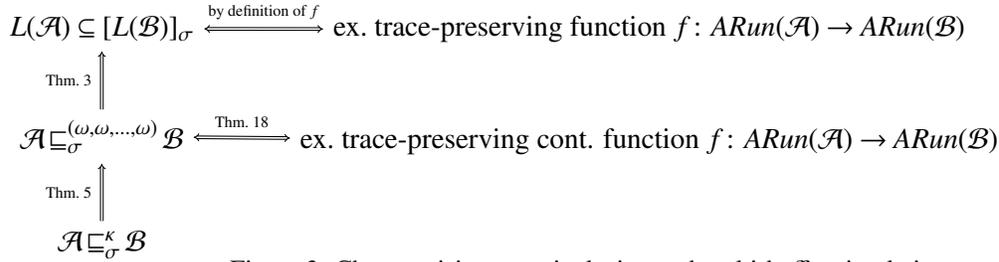
\begin{figure}[t]
	\begin{tikzpicture}[>=implies]
	\matrix (m) [matrix of math nodes, nodes in empty cells, row sep=2em, column sep=3em]
	{ L(\mathcal{A}) \subseteq [L(\mathcal{B})]_\sigma  &\text{ex.\ trace-preserving function } f\colon \ARun(\mathcal{A}) \to \ARun(\mathcal{B})     \\
		\mathcal{A} \DKsimu{(\omega,\omega,\dots,\omega)}{\sigma} \mathcal{B} &\text{ex.\ trace-preserving cont. function } f\colon \ARun(\mathcal{A}) \to \ARun(\mathcal{B})  \\
		\mathcal{A} \DKsimu{\kappa}{\sigma} \mathcal{B} & 
\\	};
	
	\draw[<->,double] (m-1-1) -- node[above] {\tiny by definition of $f$} (m-1-2);
	\draw[<->,double] (m-2-1) -- node[above] {\tiny Thm.~\ref{thm:continuous}} (m-2-2);
	\draw[<-,double] (m-1-1) -- node[left] {\tiny Thm.~\ref{thm:approximate}} (m-2-1);
	\draw[<-,double] (m-2-1) -- node[left] {\tiny Thm.~\ref{thm:hierarchies}} (m-3-1);
	\end{tikzpicture}
\caption{Characterising trace inclusion and multi-buffer simulation.}
\label{fig:summary}
\end{figure}
We have defined multi-buffer simulation relations on B\"uchi automata and analysed them with respect to their usability
for inclusion problems between trace languages defined by NBA. Fig.~\ref{fig:summary} presents a picture of how these 
concepts are related. 
There are (at least) three ways for the work presented here to be continued.

1. As can be seen from Fig.~\ref{fig:summary}, the question of whether there is a characterisation of bounded multi-buffer simulation is still open. For
single-buffer simulations, a matching criterion is known, namely that of a Lipschitz continuous function between the runs of the automata. However, it is
possible to give examples which show that Lipschitz continuity is neither sufficient nor necessary for bounded multi-buffer simulation. We suspect that
an additional condition on the looping structure of the automata in terms of the underlying dependency relation needs to be given such that Lipschitz
continuity captures bounded multi-buffer simulation.

2. In the whole of this article we have assumed that all the used
concepts like automata and games are defined w.r.t.\ a fixed trace
alphabet. It is possible to relax this and study the effect that
varying the independence relation has on the results, e.g.\ whether or
not this also induces strict hierarchies w.r.t.\ expressive
power. This could be used to refine the approximation sketched at the
end of Sect.~\ref{sec:multibuf}. This may not make sense for trace
inclusion problems but may yield better approximations for related
problems like transducer inclusion which feature a very restricted
form of independence on their alphabets.

3. Finally, recall that $\DKsimu{\omega}{\Sigma}$ is
EXPTIME-complete. There is a variant that is ``only'' PSPACE-complete
\cite{DBLP:journals/corr/HutagalungLL14}; it is obtained by requiring
\duplicator to either skip turns or flush the entire buffer. It is not
clear what the complexity of such a restricted multi-buffer game
is. Note that the undecidability proof for
$\DKsimu{(\omega,0)}{(\Sigma_1,\Sigma_2)}$ (Thm.~\ref{thm:RBG2simu})
heavily relies on \duplicator's ability to constantly keep some
content in the buffer, namely the last configuration of a Turing
machine in its simulation.

\vspace*{-5mm}
\bibliographystyle{eptcs}
\bibliography{literature}

\end{document}